\documentclass[preprint,12pt,authoryear]{elsarticle}

\usepackage{amsmath,amssymb,amsthm,mathtools}
\usepackage{enumitem}
\usepackage{booktabs}
\usepackage{setspace}
\usepackage{hyperref}
\AtBeginDocument{\hypersetup{colorlinks=true, breaklinks=true,
  linkcolor=[rgb]{0,0,0.45}, citecolor=[rgb]{0,0,0.45},
  urlcolor=[rgb]{0,0,0.45}, filecolor=[rgb]{0,0,0.45}}}

\graphicspath{{figures/}{./}}

\journal{Games and Economic Behavior}

\newtheorem{theorem}{Theorem}[section]
\newtheorem{proposition}[theorem]{Proposition}
\newtheorem{lemma}[theorem]{Lemma}
\newtheorem{corollary}[theorem]{Corollary}

\newtheorem{conjecture}[theorem]{Conjecture}
\theoremstyle{definition}
\newtheorem{definition}[theorem]{Definition}
\newtheorem{example}[theorem]{Example}
\newtheorem{assumption}[theorem]{Assumption}
\theoremstyle{remark}
\newtheorem{remark}[theorem]{Remark}

\newcommand{\R}{\mathbb{R}}
\newcommand{\N}{\mathbb{N}}
\newcommand{\calA}{\mathcal{A}}
\newcommand{\calM}{\mathcal{M}}
\newcommand{\calS}{\mathcal{S}}
\newcommand{\sgn}{\operatorname{sgn}}

\begin{document}

\begin{frontmatter}

\title{Low-Rank Payoffs and Limit Uniqueness in Global Games}

%%%%%%%%%%%%%%%%%%%%%%%%%%%%%%%%%%%%%%%%%%%%%%%%%%%%%%%%%%%%%%%%%%%%%%%%
%% CHECK BEFORE POSTING: confirm the e-mail address below is the author's
%% actual institutional address.
%%%%%%%%%%%%%%%%%%%%%%%%%%%%%%%%%%%%%%%%%%%%%%%%%%%%%%%%%%%%%%%%%%%%%%%%
\author[inst1]{Dana Golden\corref{cor1}}
\ead{dana.golden@stonybrook.edu}
\cortext[cor1]{Corresponding author.}

\affiliation[inst1]{organization={Center of Excellence in Wireless and
              Information Technology (CEWIT), Stony Brook University},
            addressline={1500 Stony Brook Road},
            city={Stony Brook},
            postcode={11794-6040},
            state={NY},
            country={USA}}

\begin{abstract}
When does the global game information structure select a unique equilibrium? Limit uniqueness in two-player supermodular games fails exactly when a risk-dominant better response cycle exists (Veiel, 2025). We show that rank-one factor structure on payoffs eliminates such cycles entirely, so every rank-one supermodular game admits a generalized ordinal potential and limit uniqueness follows for any number of actions. The boundary is sharp: an explicit three-action rank-two game carries a length-six cycle, no supermodular game carries a cycle of length four, and every game within a quantified sup-norm margin of a nondegenerate rank-one game is cycle-free. Rank-one structure can also be manufactured: when players compete across many independent markets with common latent payoffs, the stacked observation matrix is rank one plus sparse, and a Robust PCA estimator leaves residual noise that vanishes with the signal scale yet stays positive at any finite sample, even under partial observation.
\end{abstract}


\end{frontmatter}

%----------------------------------------------------------------------
\section{Introduction}
\label{sec:intro}
%----------------------------------------------------------------------

A central question in the theory of strategic interaction is when the information structure of a game selects a unique equilibrium. The global games framework, originating with \citet{carlsson1993global} and generalized by \citet{frankel2003equilibrium}, provides a powerful answer for games with strategic complementarities. When players receive slightly noisy private signals about an underlying state, the incomplete information game often has a unique equilibrium even though the complete information game has many. The economic force is contagion: extreme types, for whom one action is dominant, infect nearby types through the complementarity structure, and multiplicity unravels.

\citet{veiel2025limits} sharpens this picture for two-player, many-action games with symmetric, multi-dimensional noise. His Theorem~4.1 gives a precise characterization. For quasi-concave supermodular games, limit uniqueness holds if and only if no best-reply closed set contains a risk-dominant \emph{better response cycle} (BRC). By \citet{monderer1996potential}, the absence of BRCs is equivalent to the existence of a generalized ordinal potential. A BRC lets best-response contagion circulate instead of terminating, so dominance regions can no longer unravel inward to pin down one profile.

The obstruction to selection therefore has a name. The question this paper answers is what economic structure on payoffs removes it.

\medskip
\noindent \textbf{This paper.} We show that a low-rank factor structure on payoffs eliminates the degrees of freedom that BRCs require. When payoff matrices have rank one, each player's payoff is multiplicatively separable: an ``own-action value'' scaled by an ``opponent sensitivity.'' The opponent's action rescales the magnitude of payoff differences but never reorders them, except possibly by flipping the entire ranking when the sensitivity changes sign. Payoff gaps at any two own actions stay proportional across opponent actions. This proportionality is exactly what rules out the opposing skewness that \citet[Section~6]{veiel2025limits} identifies as necessary for BRCs.

Rank-one payoffs arise naturally when strategic interaction is mediated by a single latent factor. In an investment game, all payoff variation across action pairs might be governed by one market-conditions factor. More broadly, any one-dimensional quality index, demand shock, or cost shock that multiplicatively scales the returns to action generates rank-one payoffs. Linear-index and single-factor specifications common in empirical work do so by construction.

\medskip
\noindent \textbf{Contributions.} We make four contributions.

\emph{First}, we prove that supermodular games with rank-one payoff matrices admit no BRCs of any length (Theorem~\ref{thm:rank1}). The proof introduces a sign-quadrant decomposition of the action space, illustrated in Figure~\ref{fig:quadrants}. Supermodularity forces the rank-one factors to be monotone, so each player's improvement direction is pinned down by the sign of a single function of the opponent's action. The action space splits into four quadrants with monotone dynamics: two aligned quadrants, in which both players escalate or both retreat, that are absorbing, and two opposed quadrants that are mutually unreachable without passing through an absorbing one. A cycle must return to its starting point, so no cycle exists. Combined with Veiel's characterization and the Monderer--Shapley equivalence, this yields limit uniqueness for rank-one quasi-concave supermodular global games (Corollary~\ref{cor:limuniq}) for any number of actions.

\emph{Second}, we map the boundary of the rank-one result (Section~\ref{sec:rankr}). The result is tight: an explicit $3 \times 3$ strictly supermodular game in which both payoff matrices have rank two carries a length-six BRC (Example~\ref{ex:rank2}, Figure~\ref{fig:cycle}). Because the BRC conditions are strict inequalities, such games occupy an open, positive-measure set on the rank-two manifold. No genericity argument can substitute for the rank restriction. We also prove that no two-player supermodular game of any rank admits a BRC of length four, so six is the minimal cycle length and each player needs at least three actions (Proposition~\ref{prop:no_length4}). A margin theorem (Theorem~\ref{thm:margin}) then shows the conclusion is robust: any game within sup-norm distance $\gamma/2$ of a nondegenerate rank-one supermodular game with margin $\gamma$ is cycle-free, whatever its rank.

\emph{Third}, we show the rank-one structure can be manufactured from data rather than assumed (Section~\ref{sec:rpca}). Players compete simultaneously in $M$ independent markets governed by the same latent payoffs, subject to market-specific sparse corruptions and dense noise. Standard Robust PCA (RPCA) guarantees require large matrix dimensions, yet game-theoretic action spaces are finite and fixed. Our resolution is to let the large dimension come from the number of markets. Stacking cross-market observations produces a matrix whose low-rank component has rank exactly one regardless of the rank of the payoff matrices (Figure~\ref{fig:stacking}). A two-step estimator, entry-wise median for sparse removal followed by low-rank projection, delivers estimation error of order $O\bigl(\sigma \sigma_w (\varepsilon + \sqrt{n \log n / M})\bigr)$. The error vanishes as the noise scale $\sigma \to 0$ but stays strictly positive at any finite $M$. This is precisely the intermediate regime the contagion mechanism requires: exact recovery would return the game to complete information, where multiplicity comes back.

\emph{Fourth}, we extend the multi-market framework to partial observation (Section~\ref{sec:matcomp}), where players see only a subset of action-pair payoffs. We analyze three observation models of increasing economic realism: missing markets, missing entries under uniform random observation, and structured missingness arising endogenously from equilibrium play. In each case the estimator adapts cleanly and the intermediate regime survives, because every column of the stacked low-rank component is identical, which makes partial observation far less damaging than in generic matrix completion. The strategic observation model yields a striking complementarity: when different markets realize different equilibria, the resulting variation in observed action pairs covers the payoff space. Equilibrium multiplicity, usually a nuisance, here aids recovery. The estimation guarantees extend to all three settings under mild coverage conditions, and the selection statement extends conditionally on the same noise-class extension that underlies the multi-market conjecture (Corollary~\ref{cor:limuniq_mc}).

\medskip
\noindent \textbf{Unconditional and conditional results.} Our structural results stand on \citet{monderer1996potential} and direct arguments alone: the absence of better response cycles at rank one, the generalized ordinal potential, the sharpness and minimal-length results, and the margin theorem, as do all estimation results in Sections~\ref{sec:rpca} and~\ref{sec:matcomp}. The translation of cycle absence into limit uniqueness is conditional on Theorem~4.1 of \citet{veiel2025limits}, which at the time of writing is an unpublished working paper. The multi-market selection statement additionally requires an extension of his noise class, and we state it as Conjecture~\ref{conj:multimarket_lu} rather than as a theorem. Remark~\ref{rem:conditional} records the partition precisely.

\medskip
\noindent \textbf{Related literature.} The global games framework originates with \citet{carlsson1993global}, was applied by \citet{morris1998unique} to currency crises, and is surveyed by \citet{morris2003global}; \citet{frankel2003equilibrium} generalized it to many-action games with strategic complementarities. \citet{mathevet2010contraction} obtains selection in finite global games through a contraction principle. The robustness literature following \citet{weinstein2007structure}, and for supermodular games \citet{oyama2020generalized}, studies which equilibrium predictions survive perturbations of higher-order beliefs; the limit uniqueness question studied here is the global-game instance of that broader program. \citet{veiel2025limits} provides the characterization we build on, and the potential games literature \citep{monderer1996potential} supplies the equivalence between the absence of improvement cycles and generalized ordinal potentials.

A separate strand studies bimatrix games that are low-rank in the sense of the rank of the \emph{sum} $\Pi_1 + \Pi_2$ \citep{kannan2010rank,adsul2021rank1,vonstengel2012rank1}. That notion is genuinely different from the player-by-player rank restriction used here; Remark~\ref{rem:rank_terminology} distinguishes the two and explains why the contrast is instructive.

The RPCA methodology draws on \citet{candes2011robust} and \citet{zhou2010stable}, and the partial observation extension connects to matrix completion \citep{candes2010matrix,keshavan2010matrix}. The multi-market construction relates to the multimarket-contact literature \citep{bernheim1990multimarket}, with an inverted mechanism: rather than multimarket contact enabling collusion through repeated-game incentives, here it enables equilibrium \emph{selection} through superior information processing. The strategic observation model, in which equilibrium diversity drives payoff learning, suggests connections to information acquisition in games; we leave those formal connections for future work.

\medskip
\noindent \textbf{Outline.} Section~\ref{sec:model} introduces the model. Section~\ref{sec:rank1} presents the main result: rank-one supermodular games have no better response cycles. Section~\ref{sec:rankr} establishes tightness at rank two, a minimal cycle length of six, and a quantitative robustness theorem for near-rank-one payoffs. Section~\ref{sec:rpca} develops the multi-market RPCA framework. Section~\ref{sec:matcomp} extends it to partial observation. Section~\ref{sec:discussion} concludes with implications and open questions.

%----------------------------------------------------------------------
\section{Model}
\label{sec:model}
%----------------------------------------------------------------------

\subsection{The economic environment}

We study two-player games with strategic complementarities and finite action spaces. Fix two players $i \in \{1,2\}$ with finite, ordered action sets $A_i = \{1, \ldots, N_i\}$, where $N_i \geq 2$. A \emph{payoff matrix} assigns to each action profile $(a_1, a_2) \in A_1 \times A_2$ a payoff $\pi_i(a_1,a_2) \in \R$ for each player~$i$. We identify player~$i$'s payoffs with the matrix $\Pi_i \in \R^{N_1 \times N_2}$, where $(\Pi_i)_{a_1,a_2} = \pi_i(a_1,a_2)$.

\begin{definition}[Supermodularity]
The game $(\Pi_1, \Pi_2)$ is \emph{supermodular} if for every player~$i$, every $a_i \in \{1,\ldots,N_i-1\}$, and every $a_{-i} \in \{1,\ldots,N_{-i}-1\}$:
\[
\pi_i(a_i+1, a_{-i}+1) - \pi_i(a_i, a_{-i}+1) > \pi_i(a_i+1, a_{-i}) - \pi_i(a_i, a_{-i}).
\]
Equivalently, the marginal return to increasing one's own action is strictly increasing in the opponent's action. This is the standard complementarity condition in the games literature.
\end{definition}

\subsection{Better response cycles}

The key object connecting payoff structure to equilibrium selection is the better response cycle, which we define following \citet{veiel2025limits}.

\begin{definition}[Better response cycle]
\label{def:brc}
A \emph{better response cycle (BRC)} of length $m \geq 4$ is a sequence of action profiles
\[
(a_1^0, a_2^0),\; (a_1^1, a_2^1),\; \ldots,\; (a_1^m, a_2^m) = (a_1^0, a_2^0),
\]
and a player $i \in \{1,2\}$ such that: for every even $n \in \{0,\ldots,m-1\}$, player~$i$ strictly improves while player~$-i$'s action is held fixed:
\[
\pi_i(a_i^{n+1}, a_{-i}^n) > \pi_i(a_i^n, a_{-i}^n) \quad \text{and} \quad a_{-i}^{n+1} = a_{-i}^n;
\]
and for every odd $n \in \{0,\ldots,m-1\}$, player~$-i$ strictly improves while player~$i$'s action is held fixed:
\[
\pi_{-i}(a_i^n, a_{-i}^{n+1}) > \pi_{-i}(a_i^n, a_{-i}^n) \quad \text{and} \quad a_i^{n+1} = a_i^n.
\]
The minimum possible cycle length is $m = 4$, corresponding to a ``rectangular'' cycle in which players alternate single moves around four action profiles.
\end{definition}

Figure~\ref{fig:brc_schematic} depicts the object schematically. Each arrow is a strictly improving unilateral deviation, the movers alternate, and the sequence closes on itself. It is exactly this closed loop that lets best-response contagion circulate rather than terminate.

\begin{figure}[t]
\centering
\includegraphics[width=0.42\textwidth]{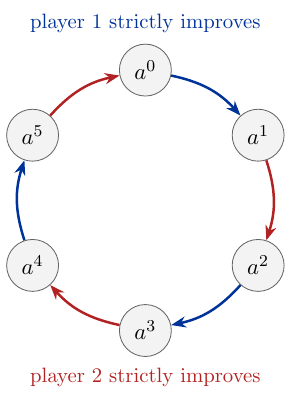}
\caption{\textbf{A better response cycle.} The object of Definition~\ref{def:brc}, drawn schematically for length six. Each node is an action profile $a^n = (a_1^n, a_2^n)$; each arrow is a strictly improving deviation by one player with the other's action held fixed, and the movers alternate. Because the sequence returns to its starting profile, best-response contagion can circulate around the loop indefinitely instead of terminating, which is the mechanism by which BRCs obstruct limit uniqueness. Section~\ref{sec:rankr} shows that in supermodular games six is in fact the minimal length, and Figure~\ref{fig:cycle} realizes this schematic in an explicit rank-two game.}
\label{fig:brc_schematic}
\end{figure}

Two results in the literature establish the role of BRCs in selection. \citet{monderer1996potential} show that a game admits a generalized ordinal potential on a set of actions if and only if there are no better response cycles on that set. \citet{veiel2025limits} shows that limit uniqueness in global games holds if and only if no BRC set contains a risk-dominant better response cycle. Our task is therefore to identify conditions on payoffs under which BRCs cannot arise.

Definition~\ref{def:brc} requires the improving player to alternate. The Monderer--Shapley equivalence is stated for improvement cycles with an arbitrary order of movers. The following reduction shows that, in two-player games, nothing is lost by restricting attention to alternating cycles.

\begin{lemma}[Reduction to alternating cycles]
\label{lem:alternation}
Suppose a two-player game admits an improvement cycle: a cyclic sequence of action profiles in which consecutive profiles differ in exactly one player's action and the deviating player strictly improves, with no restriction on the order of movers. Then the game admits a BRC in the sense of Definition~\ref{def:brc}.
\end{lemma}

\begin{proof}
Consecutive moves by the same player compose. If player~$i$ moves from $a_i$ to $a_i'$ and then to $a_i''$ against a fixed $a_{-i}$, with each step strictly improving, then $\pi_i(a_i'', a_{-i}) > \pi_i(a_i, a_{-i})$, so the two steps may be replaced by one. Iterating, every improvement cycle reduces to one in which the movers strictly alternate. A cycle in which only one player ever moves is impossible: that player's payoff strictly increases along the cycle and cannot return to its initial value. The reduced alternating cycle therefore has even length at least four.
\end{proof}

\subsection{Low-rank payoff structure}

We now introduce the structural restriction on payoffs that drives our results.

\begin{definition}[Rank-$r$ payoffs]
\label{def:rank}
Player~$i$'s payoff matrix $\Pi_i$ has \emph{rank~$r$} if $\operatorname{rank}(\Pi_i) = r$. Equivalently, $\Pi_i$ admits a decomposition
\[
\Pi_i = \sum_{\ell=1}^{r} \sigma_\ell^{(i)}\, \mathbf{u}_\ell^{(i)} (\mathbf{v}_\ell^{(i)})^\top,
\]
where $\sigma_\ell^{(i)} > 0$ and $\mathbf{u}_\ell^{(i)} \in \R^{N_1}$, $\mathbf{v}_\ell^{(i)} \in \R^{N_2}$ are the singular vectors.
\end{definition}

When $r = 1$, the payoff matrix is \emph{multiplicatively separable}:
\begin{equation}
\label{eq:rank1}
\pi_i(a_1, a_2) = f_i(a_1)\, g_i(a_2),
\end{equation}
for functions $f_i \colon A_1 \to \R$ and $g_i \colon A_2 \to \R$ (absorbing the singular value into either factor). Each player's payoff factors into an own-action value $f_i(a_i)$ modulated by an opponent sensitivity $g_i(a_{-i})$. As discussed in the introduction, the opponent's action can scale the returns to one's own action, but cannot reorder them except by flipping the whole ranking when the sensitivity changes sign. This is the structural feature that prevents better response cycles.

\begin{remark}[Two notions of rank for bimatrix games]
\label{rem:rank_terminology}
The restriction above applies to each player's payoff matrix \emph{separately}. It should not be confused with the rank of a bimatrix game in the sense of \citet{kannan2010rank}, which is defined as $\operatorname{rank}(\Pi_1 + \Pi_2)$, so that rank-zero games are zero-sum games and the rank grades the departure from the zero-sum benchmark. The two restrictions are genuinely different and neither implies the other. A game with $\operatorname{rank}(\Pi_1 + \Pi_2) = 1$ may have individual payoff matrices of full rank, and two rank-one payoff matrices generically sum to a rank-two matrix.

The two literatures also ask different questions. The fixed-rank hierarchy of \citet{kannan2010rank} concerns the computation of Nash equilibria, and \citet{adsul2021rank1} give a polynomial-time algorithm for finding an equilibrium of a rank-one game in the sum sense. Most instructive as a contrast, \citet{vonstengel2012rank1} constructs games of sum-rank one with exponentially many Nash equilibria. Low rank of the sum therefore places essentially no restriction on the size of the equilibrium set. Rank one of each payoff matrix, by contrast, combined with supermodularity and the global game information structure, selects a \emph{unique} outcome in the limit (Corollary~\ref{cor:limuniq}). Whenever we write ``rank'' without qualification, we mean the player-by-player notion of Definition~\ref{def:rank}; wherever the sum notion appears, we say so explicitly.
\end{remark}

%----------------------------------------------------------------------
\section{Main result: rank-one games have no better response cycles}
\label{sec:rank1}
%----------------------------------------------------------------------

This section contains our main result. We proceed in three steps. First, we show that supermodularity forces the rank-one factors to be monotone. Second, we characterize improvement directions and the resulting sign-quadrant dynamics. Third, we prove that no better response cycle survives these dynamics.

\subsection{Supermodularity forces monotonicity}

The first step establishes that the multiplicative factors in rank-one payoffs must be well-behaved.

\begin{lemma}[Monotonicity of rank-one factors]
\label{lem:monotone}
Let $(\Pi_1, \Pi_2)$ be a supermodular game in which each $\Pi_i$ has rank one, so that $\pi_i(a_1,a_2) = f_i(a_1) g_i(a_2)$. Then for each player~$i$, the functions $f_i$ and $g_i$ are both strictly monotone.

Moreover, after a suitable normalization that does not change payoffs (replacing $f_i \to -f_i$ and $g_i \to -g_i$ simultaneously), we may assume without loss of generality that both $f_i$ and $g_i$ are \emph{strictly increasing}.
\end{lemma}

\begin{proof}
Fix player~$i$ and consider the supermodularity condition. Writing $\Delta f_i(a_i) := f_i(a_i+1) - f_i(a_i)$ and $\Delta g_i(a_{-i}) := g_i(a_{-i}+1) - g_i(a_{-i})$, supermodularity requires:
\begin{align}
\label{eq:supermod_rank1}
&\pi_i(a_i{+}1, a_{-i}{+}1) - \pi_i(a_i, a_{-i}{+}1) > \pi_i(a_i{+}1, a_{-i}) - \pi_i(a_i, a_{-i}) \notag \\
&\quad\iff\quad \Delta f_i(a_i) \cdot g_i(a_{-i}{+}1) > \Delta f_i(a_i) \cdot g_i(a_{-i}) \notag \\
&\quad\iff\quad \Delta f_i(a_i) \cdot \Delta g_i(a_{-i}) > 0,
\end{align}
for all $a_i \in \{1,\ldots,N_i-1\}$ and $a_{-i} \in \{1,\ldots,N_{-i}-1\}$.

Since~\eqref{eq:supermod_rank1} must hold for \emph{all} pairs $(a_i, a_{-i})$, the signs $\sgn(\Delta f_i(a_i))$ and $\sgn(\Delta g_i(a_{-i}))$ must be constant and equal. Therefore, $f_i$ and $g_i$ are both strictly monotone in the same direction.

If both are strictly decreasing, replace $f_i \to -f_i$ and $g_i \to -g_i$. Since $(-f_i)(-g_i) = f_i g_i$, the payoffs are unchanged, and both new functions are strictly increasing.
\end{proof}

Economically, higher actions yield higher own-action value and higher opponent actions yield higher sensitivity, or vice versa after relabeling.

\subsection{Improvement directions under rank one}

Throughout this subsection, fix a supermodular game with rank-one payoffs and assume (by Lemma~\ref{lem:monotone}) that $f_1, g_1, f_2, g_2$ are all strictly increasing.

\begin{lemma}[Improvement direction]
\label{lem:direction}
For player~$1$, the payoff difference from switching action $a_1$ to $a_1'$ while player~$2$ plays $a_2$ is:
\[
\pi_1(a_1', a_2) - \pi_1(a_1, a_2) = \bigl[f_1(a_1') - f_1(a_1)\bigr] \cdot g_1(a_2).
\]
Since $f_1$ is strictly increasing, $\sgn\bigl(f_1(a_1') - f_1(a_1)\bigr) = \sgn(a_1' - a_1)$. Therefore:
\begin{enumerate}[label=(\roman*),nosep]
\item If $g_1(a_2) > 0$: player~$1$ improves by increasing $a_1$.
\item If $g_1(a_2) < 0$: player~$1$ improves by decreasing $a_1$.
\item If $g_1(a_2) = 0$: all payoff differences are zero; no strict improvement is possible.
\end{enumerate}
Analogously, player~$2$ improves by increasing $a_2$ when $f_2(a_1) > 0$, and by decreasing $a_2$ when $f_2(a_1) < 0$.
\end{lemma}

\begin{proof}
Direct computation using $\pi_j(a_1,a_2) = f_j(a_1) g_j(a_2)$.
\end{proof}

The following observation handles the boundary cases in which a sensitivity function vanishes. It lets the cycle analysis below treat the zero-sign states explicitly rather than implicitly.

\begin{lemma}[No indifference at the mover's turn]
\label{lem:nodegenerate}
Along any BRC, with step indices taken modulo $m$: $g_1(a_2^n) \neq 0$ at every even step $n$, and $f_2(a_1^n) \neq 0$ at every odd step $n$.
\end{lemma}

\begin{proof}
At an even step, player~$1$ must strictly improve against $a_2^n$. If $g_1(a_2^n) = 0$, every available payoff difference $[f_1(a_1') - f_1(a_1)]\, g_1(a_2^n)$ vanishes and no strict improvement exists. The odd case is symmetric. Because the cycle is closed, the requirement applies at every index modulo $m$.
\end{proof}

The economic content of Lemma~\ref{lem:direction} is that rank-one payoffs create a particularly stark form of strategic complementarity. Each player's entire improvement direction is pinned down by a single sign: the sign of the opponent-sensitivity function at the opponent's current action. In general supermodular games, by contrast, the improvement direction can depend on the action profile in a more complex way.

\subsection{Sign quadrants and monotone dynamics}

Since $g_1$ and $f_2$ are strictly increasing functions on finite domains, each changes sign at most once. This splits the action grid into four quadrants with qualitatively distinct dynamics, illustrated in Figure~\ref{fig:quadrants}.

\begin{definition}[Threshold actions]
\label{def:thresholds}
Define:
\begin{align*}
a_2^* &:= \max\{a_2 \in A_2 : g_1(a_2) \leq 0\}, \\
a_1^* &:= \max\{a_1 \in A_1 : f_2(a_1) \leq 0\},
\end{align*}
with $a_2^* = 0$ if $g_1(a_2) > 0$ for all $a_2$, and $a_2^* = N_2$ if $g_1(a_2) \leq 0$ for all $a_2$. Analogously for $a_1^*$.
\end{definition}

\begin{definition}[Sign quadrants]
\label{def:quadrants}
For $(a_1, a_2) \in A_1 \times A_2$, define the \emph{sign quadrant}:
\[
Q(a_1, a_2) := \bigl(\sgn(f_2(a_1)),\; \sgn(g_1(a_2))\bigr) \in \{+,-,0\}^2.
\]
Write $Q^{++}, Q^{+-}, Q^{-+}, Q^{--}$ for the four strict quadrants (both signs nonzero). Table~\ref{tab:quadrants} records the improvement directions in each quadrant.
\end{definition}

\begin{table}[t]
\centering
\caption{\textbf{Improvement directions by sign quadrant.} Each player's improving direction under rank-one supermodular payoffs, as a function of the signs $\sgn(f_2(a_1))$ and $\sgn(g_1(a_2))$ that define the quadrant (Definition~\ref{def:quadrants}). ``Up'' means the mover strictly raises her own action and ``Down'' that she strictly lowers it. Directions are aligned in $Q^{++}$ and $Q^{--}$ and opposed in $Q^{+-}$ and $Q^{-+}$.}
\label{tab:quadrants}
\begin{tabular}{lll}
\toprule
Quadrant & Player 1 direction & Player 2 direction \\
\midrule
$Q^{++}$ & Up ($a_1$ increases)   & Up ($a_2$ increases)   \\
$Q^{+-}$ & Down ($a_1$ decreases) & Up ($a_2$ increases)   \\
$Q^{-+}$ & Up ($a_1$ increases)   & Down ($a_2$ decreases) \\
$Q^{--}$ & Down ($a_1$ decreases) & Down ($a_2$ decreases) \\
\bottomrule
\end{tabular}
\end{table}

\begin{figure}[t]
\centering
\includegraphics[width=0.62\textwidth]{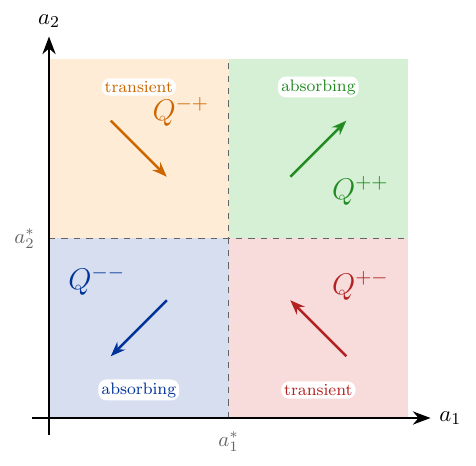}
\caption{\textbf{The sign-quadrant decomposition of the action grid.} The dashed lines mark the threshold actions $a_1^*$ and $a_2^*$ at which the sensitivity functions $f_2$ and $g_1$ change sign. Arrows show each quadrant's joint improvement direction. In the aligned quadrants $Q^{++}$ and $Q^{--}$ both players push the same way, and Proposition~\ref{prop:absorbing} shows these regions are absorbing. In the misaligned quadrants $Q^{+-}$ and $Q^{-+}$ the players pull in opposite directions; Proposition~\ref{prop:no_cross} shows every exit from a misaligned quadrant lands in an absorbing region, so the two misaligned quadrants are mutually unreachable. A cycle would need to revisit a quadrant, but the dynamics never allow it.}
\label{fig:quadrants}
\end{figure}

The decomposition captures the following intuition. In $Q^{++}$, both players face positive opponent sensitivity, so both wish to escalate: a regime of mutual reinforcement. In $Q^{--}$, both face negative sensitivity, so both wish to de-escalate: a regime of mutual retreat. In the off-diagonal quadrants $Q^{+-}$ and $Q^{-+}$, the players pull in opposite directions. The key insight, visible in Figure~\ref{fig:quadrants}, is that the aligned quadrants are absorbing while the opposed quadrants are mutually unreachable.

\begin{lemma}[Quadrant transition rules]
\label{lem:transitions}
Consider a BRC with player~$1$ moving at even steps and player~$2$ at odd steps. The following transition constraints hold at each step:

\begin{enumerate}[label=(\roman*),nosep]
\item \textbf{Player~$1$ moves (even $n$), $g_1(a_2^n) > 0$:} Player~$1$ increases $a_1$. Since $f_2$ is strictly increasing, the first coordinate of $Q$ is non-decreasing: $\sgn(f_2(a_1^{n+1})) \geq \sgn(f_2(a_1^n))$. The second coordinate is unchanged.

\item \textbf{Player~$1$ moves (even $n$), $g_1(a_2^n) < 0$:} Player~$1$ decreases $a_1$. The first coordinate of $Q$ is non-increasing. The second coordinate is unchanged.

\item \textbf{Player~$2$ moves (odd $n$), $f_2(a_1^n) > 0$:} Player~$2$ increases $a_2$. The second coordinate of $Q$ is non-decreasing. The first coordinate is unchanged.

\item \textbf{Player~$2$ moves (odd $n$), $f_2(a_1^n) < 0$:} Player~$2$ decreases $a_2$. The second coordinate of $Q$ is non-increasing. The first coordinate is unchanged.
\end{enumerate}
\end{lemma}

\begin{proof}
We prove~(i); the others are analogous. At even step $n$ with $g_1(a_2^n) > 0$, Lemma~\ref{lem:direction} implies player~$1$ increases $a_1$: $a_1^{n+1} > a_1^n$. Since $f_2$ is strictly increasing, $f_2(a_1^{n+1}) > f_2(a_1^n)$. If $f_2(a_1^n) > 0$ then $f_2(a_1^{n+1}) > 0$; if $f_2(a_1^n) < 0$ then $f_2(a_1^{n+1})$ could be negative, zero, or positive. In all cases, $\sgn(f_2(a_1^{n+1})) \geq \sgn(f_2(a_1^n))$. The second coordinate $\sgn(g_1(a_2^{n+1})) = \sgn(g_1(a_2^n))$ because $a_2^{n+1} = a_2^n$.
\end{proof}

We now establish the two key structural properties of the sign-quadrant dynamics.

\begin{proposition}[The aligned regions are absorbing]
\label{prop:absorbing}
Define the aligned regions
\begin{align*}
U &:= \{(a_1,a_2) : g_1(a_2) > 0,\; f_2(a_1) \geq 0\}, \\
D &:= \{(a_1,a_2) : g_1(a_2) < 0,\; f_2(a_1) \leq 0\}.
\end{align*}
Note $Q^{++} \subseteq U$ and $Q^{--} \subseteq D$. Along any BRC, once the trajectory enters $U$ (respectively $D$) it remains there permanently. Moreover, while the trajectory is in $U$, every player-$1$ move strictly increases $a_1$; while in $D$, every player-$1$ move strictly decreases $a_1$.
\end{proposition}

\begin{proof}
Suppose $(a_1^n, a_2^n) \in U$, so $g_1(a_2^n) > 0$ and $f_2(a_1^n) \geq 0$.

\emph{Case 1: $n$ is even} (player~$1$ moves). Since $g_1(a_2^n) > 0$, Lemma~\ref{lem:direction} forces $a_1^{n+1} > a_1^n$. Since $f_2$ is strictly increasing, $f_2(a_1^{n+1}) > f_2(a_1^n) \geq 0$. The second coordinate is unchanged: $g_1(a_2^{n+1}) = g_1(a_2^n) > 0$. Hence the trajectory stays in $U$, and the player-$1$ move was strictly upward.

\emph{Case 2: $n$ is odd} (player~$2$ moves). By Lemma~\ref{lem:nodegenerate}, $f_2(a_1^n) \neq 0$, so $f_2(a_1^n) > 0$, and Lemma~\ref{lem:direction} forces $a_2^{n+1} > a_2^n$. Since $g_1$ is strictly increasing, $g_1(a_2^{n+1}) > g_1(a_2^n) > 0$. The first coordinate is unchanged. Hence the trajectory stays in $U$.

The argument for $D$ is symmetric.
\end{proof}

\begin{proposition}[Off-diagonal quadrants are trapping outside the aligned regions]
\label{prop:no_cross}
From a profile in $Q^{+-}$, a single BRC step leads to a profile in $Q^{+-} \cup U \cup D$, where $U, D$ are the aligned regions of Proposition~\ref{prop:absorbing}. Symmetrically, from a profile in $Q^{-+}$, a single step leads to $Q^{-+} \cup U \cup D$. In particular, a BRC trajectory that never visits $U \cup D$ remains in a single off-diagonal quadrant.
\end{proposition}

\begin{proof}
Consider a profile in $Q^{+-}$: $f_2(a_1^n) > 0$ and $g_1(a_2^n) < 0$.

At an \emph{even step} (player~$1$ moves): $g_1(a_2^n) < 0$ implies player~$1$ decreases $a_1$, so $f_2(a_1^{n+1}) < f_2(a_1^n)$ while $g_1$ is unchanged at a negative value. The destination has $g_1 < 0$ and $f_2$ of any sign: it lies in $Q^{+-}$ (if $f_2 > 0$) or in $D$ (if $f_2 \leq 0$).

At an \emph{odd step} (player~$2$ moves): $f_2(a_1^n) > 0$ implies player~$2$ increases $a_2$, so $g_1(a_2^{n+1}) > g_1(a_2^n)$ while $f_2$ is unchanged at a positive value. The destination has $f_2 > 0$ and $g_1$ of any sign. If $g_1(a_2^{n+1}) < 0$ it lies in $Q^{+-}$; if $g_1(a_2^{n+1}) > 0$ it lies in $U$. The remaining possibility, $g_1(a_2^{n+1}) = 0$, is excluded: step $n+1$ is even, and Lemma~\ref{lem:nodegenerate} requires $g_1(a_2^{n+1}) \neq 0$ at even steps.

The argument for $Q^{-+}$ is symmetric.
\end{proof}

\subsection{The main theorem}

We now state and prove the main result.

\begin{theorem}[Rank-one BRC avoidance]
\label{thm:rank1}
Let $(\Pi_1, \Pi_2)$ be a supermodular game in which each player's payoff matrix has rank one. Then there are no better response cycles.
\end{theorem}

\begin{proof}
Suppose for contradiction that a BRC of length $m \geq 4$ exists. By Lemma~\ref{lem:monotone}, assume without loss of generality that $f_1, g_1, f_2, g_2$ are all strictly increasing.

\medskip
\noindent\textbf{Step 1: The cycle avoids the aligned regions $U$ and $D$.}
By Proposition~\ref{prop:absorbing}, $U$ and $D$ are forward invariant. Because the trajectory is cyclic, if any visited profile lies in $U$, then every visited profile does. But inside $U$, every player-$1$ move strictly increases $a_1$ (Proposition~\ref{prop:absorbing}), so the even-step subsequence $(a_1^n)$ is strictly increasing and cannot return to its initial value. The same argument applies to $D$ with strictly decreasing $a_1$. Hence no visited profile lies in $U \cup D$.

\medskip
\noindent\textbf{Step 2: Every visited profile lies in $Q^{+-} \cup Q^{-+}$.}
Classify a visited profile by the sign pair $(\sgn(f_2(a_1^n)), \sgn(g_1(a_2^n)))$. By Lemma~\ref{lem:nodegenerate}, the second coordinate is nonzero at even steps and the first is nonzero at odd steps. States with $g_1 > 0$ and $f_2 \geq 0$ lie in $U$; states with $g_1 < 0$ and $f_2 \leq 0$ lie in $D$; both are excluded by Step~1. A state with first coordinate $0$ at an odd step, or second coordinate $0$ at an even step, is excluded by Lemma~\ref{lem:nodegenerate}. Two cases remain. A state $(+, 0)$ can occur only at an odd step; player~$2$ then moves up (since $f_2 > 0$), making $g_1(a_2^{n+1}) > 0$, so the next profile lies in $U$, contradicting Step~1. Symmetrically, a state $(-, 0)$ at an odd step leads into $D$. Every remaining state is $(+,-)$ or $(-,+)$; that is, every visited profile lies in $Q^{+-}$ or $Q^{-+}$.

\medskip
\noindent\textbf{Step 3: The cycle remains in a single off-diagonal quadrant.}
By Proposition~\ref{prop:no_cross}, a step from $Q^{+-}$ lands in $Q^{+-} \cup U \cup D$, and a step from $Q^{-+}$ lands in $Q^{-+} \cup U \cup D$. Since $U \cup D$ is never visited (Step~1), the trajectory stays in whichever off-diagonal quadrant it starts in.

\medskip
\noindent\textbf{Step 4: Monotonicity within a single off-diagonal quadrant prevents cycling.}
Suppose the entire trajectory stays in $Q^{+-}$. At every even step, $g_1(a_2^n) < 0$, so player~$1$ strictly decreases $a_1$; the even-step subsequence $(a_1^n)$ is strictly decreasing and cannot return to its initial value, contradicting closure of the cycle. The case $Q^{-+}$ is symmetric, with player~$1$ strictly increasing.

\medskip
Since every case leads to a contradiction, no BRC exists.
\end{proof}

Two features of this result are worth noting. First, the proof does not require quasi-concavity. The absence of BRCs follows from supermodularity and rank one alone. Quasi-concavity enters only through Veiel's Theorem~4.1, which requires it for the class $\calS^*$ of payoffs to which the limit uniqueness characterization applies. Second, the boundary cases $g_1(a_2) = 0$ or $f_2(a_1) = 0$ are handled explicitly: by Lemma~\ref{lem:nodegenerate}, the mover's sensitivity never vanishes at the mover's turn, and Step~2 of the proof shows that a profile at which the \emph{other} player's sensitivity vanishes feeds the trajectory into an aligned region within one step.

\subsection{Consequences for equilibrium selection}

The main theorem yields immediate consequences for potential games and global game equilibrium selection.

\begin{corollary}[Rank-one generalized ordinal potential]
\label{cor:potential}
Every supermodular game with rank-one payoff matrices admits a generalized ordinal potential on every subset of actions $B_1 \times B_2 \subseteq A_1 \times A_2$.
\end{corollary}

\begin{proof}
By \citet{monderer1996potential}, a game admits a generalized ordinal potential on $B_1 \times B_2$ if and only if there are no improvement cycles in $B_1 \times B_2$, where the order of movers is unrestricted. By Lemma~\ref{lem:alternation}, any such improvement cycle yields an alternating BRC. Since $B_1 \times B_2$ inherits rank-one payoffs and supermodularity from the full game, Theorem~\ref{thm:rank1} rules out alternating BRCs, and hence all improvement cycles.
\end{proof}

\begin{corollary}[Limit uniqueness for rank-one global games]
\label{cor:limuniq}
Let $s \in \calS^*$ be a quasi-concave supermodular payoff matrix on the unit sphere such that each player's payoff matrix has rank one. Then limit uniqueness holds at $s$ in the sense of \citet{veiel2025limits}, Theorem~4.1.
\end{corollary}

\begin{proof}
By Corollary~\ref{cor:potential}, no BRC set contains a risk-dominant better response cycle (since no better response cycles exist at all). Veiel's Theorem~4.1 then implies limit uniqueness.
\end{proof}

The logical chain connecting our results to limit uniqueness is:
\begin{multline*}
\text{Rank 1} \;\xRightarrow{\text{Thm~\ref{thm:rank1}}}\; \text{No BRCs}
\;\xRightarrow[\text{(1996)}]{\text{M\&S}}\; \text{Gen.\ ord.\ potential} \\
\;\xRightarrow[\text{Thm~4.1}]{\text{Veiel}}\; \text{Limit uniqueness.}
\end{multline*}
The conclusion holds for any action-space size, with no genericity or dimension assumptions.

\begin{remark}[Dependence on \citet{veiel2025limits}]
\label{rem:conditional}
Since Theorem~4.1 of \citet{veiel2025limits} is, at the time of writing, an unpublished working paper, it is worth stating exactly which of our results rest on it. Theorem~\ref{thm:rank1} and Corollary~\ref{cor:potential} are unconditional: they rest on \citet{monderer1996potential} and direct arguments alone. The same is true of every result in Section~\ref{sec:rankr} (Remark~\ref{rem:open}, Proposition~\ref{prop:no_length4}, Example~\ref{ex:rank2}, Theorem~\ref{thm:margin}, and its corollaries) and of every estimation result in Sections~\ref{sec:rpca} and~\ref{sec:matcomp}. Corollary~\ref{cor:limuniq} is conditional on Veiel's Theorem~4.1 exactly as stated in his paper: the signal structure there lies within his noise class, so no extension is required. The multi-market selection statement requires, in addition, an extension of his noise class to the effective noise generated by the median estimator; we therefore state it as Conjecture~\ref{conj:multimarket_lu} rather than as a theorem, with the proved content isolated in Proposition~\ref{prop:signal_properties}.
\end{remark}

%----------------------------------------------------------------------
\section{Sharpness: rank-two cycles, minimal length, and a margin theorem}
\label{sec:rankr}
%----------------------------------------------------------------------

For $r \geq 2$, the sign-quadrant argument of Section~\ref{sec:rank1} breaks down: the improvement direction at $(a_1, a_2)$ depends on $a_{-i}$ through a sum of $r$ terms with potentially varying signs. This section maps the boundary of the rank-one theorem with three results.

First, the theorem is tight. Rank two already admits better response cycles, and we give an explicit three-action example (Figure~\ref{fig:cycle}). Because BRC existence is an open condition on payoffs, no genericity argument can rescue BRC avoidance for $r \geq 2$. Second, supermodularity alone rules out the shortest cycles: no two-player supermodular game admits a BRC of length four, so every supermodular BRC has length at least six and requires at least three actions per player. Third, the rank-one conclusion is quantitatively robust: every game within an explicit sup-norm margin of a nondegenerate rank-one supermodular game is cycle-free, regardless of its own rank.

\subsection{Why genericity cannot substitute for structure}
\label{subsec:openness}

\begin{remark}[BRC existence is an open condition]
\label{rem:open}
Fix the action sets and a candidate cycle specification. The set of payoff pairs $(\Pi_1, \Pi_2)$ satisfying the cycle's strict improvement inequalities is open, both in $\R^{N_1 \times N_2} \times \R^{N_1 \times N_2}$ and relative to any rank manifold $\calM_r$. Consequently, if some game on $\calM_r$ admits a BRC, an open set of games on $\calM_r$ does, and games with BRCs have positive measure on $\calM_r$. A transversality or dimension-counting argument against the locus where all of the cycle's payoff differences vanish simultaneously is uninformative here: the topological boundary of the inequality region is contained in the \emph{union} of the single-constraint zero sets, not their intersection, so the codimension of the joint locus places no restriction on the inequality region. Genericity arguments therefore cannot deliver BRC avoidance for $r \geq 2$, and Example~\ref{ex:rank2} below shows that no such avoidance holds.
\end{remark}

\subsection{No supermodular cycle has length four}
\label{subsec:length4}

Although rank restrictions beyond one cannot eliminate cycles, supermodularity by itself constrains their geometry.

\begin{proposition}[Minimal cycle length]
\label{prop:no_length4}
No two-player supermodular game, of any rank, admits a BRC of length $m = 4$. Consequently, every BRC in a supermodular game has length at least six, and a length-six BRC requires $N_1 \geq 3$ and $N_2 \geq 3$.
\end{proposition}

\begin{proof}
A length-four cycle visits the profiles $(a,b), (a',b), (a',b'), (a,b'), (a,b)$ with
\begin{gather*}
\pi_1(a',b) > \pi_1(a,b), \qquad \pi_2(a',b') > \pi_2(a',b), \\
\pi_1(a,b') > \pi_1(a',b'), \qquad \pi_2(a,b) > \pi_2(a,b').
\end{gather*}
Relabeling both players' action sets in reverse order if necessary, which preserves supermodularity, assume $a' > a$. Since player~$2$ moves, $b' \neq b$.

If $b' > b$, telescoping the strict supermodularity inequalities over the rectangle $[a, a'] \times [b, b']$ gives
\[
\pi_1(a',b') - \pi_1(a,b') \;>\; \pi_1(a',b) - \pi_1(a,b) \;>\; 0,
\]
contradicting the third displayed inequality.

If $b' < b$, the same telescoping applied to player~$2$'s payoffs over $[a, a'] \times [b', b]$ gives
\[
\pi_2(a',b) - \pi_2(a',b') \;>\; \pi_2(a,b) - \pi_2(a,b') \;>\; 0,
\]
contradicting the second displayed inequality.

For the final claim, observe that in a length-six BRC each player makes three moves that must return to the starting action. In a two-element action set every move switches the action, so three moves end at the other action. Hence $N_1, N_2 \geq 3$.
\end{proof}

\begin{remark}[Numerical evidence on longer cycles]
\label{rem:lp_lengths}
For fixed action sets, the existence of a supermodular game supporting a given cycle specification is a linear feasibility problem in the payoff entries: both the cycle inequalities and strict supermodularity are linear. Solving these linear programs over all cycle shapes shows that length-six cycles are feasible at $3 \times 3$ (Example~\ref{ex:rank2}) and length-eight cycles at $4 \times 4$, while length-eight cycles are infeasible at $3 \times 3$ and $3 \times 4$, and length-ten cycles are infeasible at $4 \times 4$. This pattern suggests the conjecture that a supermodular BRC in which each player makes $T$ moves requires $\min(N_1, N_2) \geq T$. We leave the conjecture open.
\end{remark}

\subsection{Tightness: rank two admits cycles}
\label{subsec:tightness}

\begin{example}[A rank-two supermodular game with a BRC]
\label{ex:rank2}
Let $N_1 = N_2 = 3$ and define $\Pi_i = F_i G_i^\top$ (so each $\Pi_i$ has rank at most two) with
\begin{gather*}
F_1 = \begin{pmatrix} -0.576 & -0.436 \\ -0.718 & -0.273 \\ -4.103 & \phantom{-}0.109 \end{pmatrix}, \qquad
G_1 = \begin{pmatrix} \phantom{-}0.983 & -0.179 \\ -0.156 & -0.933 \\ -0.098 & -0.313 \end{pmatrix}, \\[4pt]
F_2 = \begin{pmatrix} \phantom{-}0.322 & \phantom{-}1.705 \\ -4.620 & \phantom{-}0.386 \\ -4.290 & -0.287 \end{pmatrix}, \qquad
G_2 = \begin{pmatrix} -0.197 & \phantom{-}0.924 \\ -0.044 & \phantom{-}0.325 \\ -0.979 & -0.200 \end{pmatrix}.
\end{gather*}
Direct computation verifies that all eight double differences of $\Pi_1$ and $\Pi_2$ are strictly positive (the game is strictly supermodular, with minimum double difference approximately $0.034$), that both matrices have rank exactly two, and that the length-six alternating sequence
\[
(3,1) \to (2,1) \to (2,3) \to (1,3) \to (1,2) \to (3,2) \to (3,1)
\]
is a BRC: each of the six moves yields a strict payoff improvement for the moving player, with the smallest improvement margin approximately $0.034$. (Actions are indexed $1$ to $3$; player~$1$ moves first.) Figure~\ref{fig:cycle} traces the cycle on the action grid.
\end{example}

\begin{figure}[t]
\centering
\includegraphics[width=0.68\textwidth]{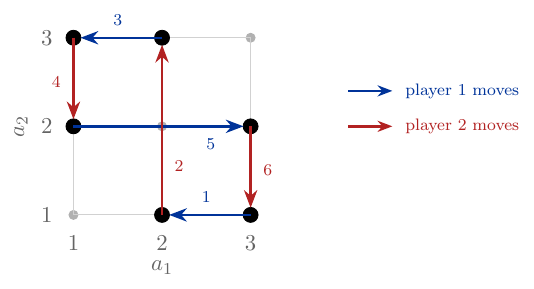}
\caption{\textbf{A length-six better response cycle at rank two.} The cycle of Example~\ref{ex:rank2} traced on the $3 \times 3$ action grid: $(3,1) \to (2,1) \to (2,3) \to (1,3) \to (1,2) \to (3,2) \to (3,1)$. Horizontal arrows are strict improvements by player~1 and vertical arrows by player~2, with steps numbered in order. Both payoff matrices have rank exactly two and the game is strictly supermodular, so the rank-one hypothesis of Theorem~\ref{thm:rank1} cannot be relaxed. All fourteen defining inequalities hold with margin at least $0.034$, so an open neighborhood of the example on the rank-two manifold consists of supermodular games with a BRC.}
\label{fig:cycle}
\end{figure}

The example establishes two things. Theorem~\ref{thm:rank1} is tight: the rank-one hypothesis cannot be relaxed to rank two, even maintaining strict supermodularity and a minimal action space. And because all fourteen inequalities hold with strictly positive margins, the example is robust: an open neighborhood of it on the rank-two manifold consists entirely of supermodular games with a BRC. Combined with Remark~\ref{rem:open}, games with BRCs occupy a positive-measure subset of rank-two supermodular games.

\subsection{Robustness: a margin theorem for near-rank-one games}
\label{subsec:margin}

The rank-one result extends quantitatively to all games sufficiently close to a well-behaved rank-one game. This requires excluding rank-one games with indifference rows or columns, near which arbitrarily small perturbations can create arbitrary preferences.

\begin{definition}[Nondegeneracy and margin]
\label{def:margin}
A rank-one supermodular game $(\Pi_1, \Pi_2)$ with $\pi_i = f_i g_i$ is \emph{nondegenerate} if $g_1(a_2) \neq 0$ for all $a_2$ and $f_2(a_1) \neq 0$ for all $a_1$. Its \emph{margin} is
\[
\gamma := \min\Bigl\{ \min_{a_1 \neq a_1',\, a_2} \bigl| [f_1(a_1') - f_1(a_1)]\, g_1(a_2) \bigr|, \;\; \min_{a_2 \neq a_2',\, a_1} \bigl| f_2(a_1)\, [g_2(a_2') - g_2(a_2)] \bigr| \Bigr\}.
\]
By Lemma~\ref{lem:monotone} the increments of $f_1$ and $g_2$ are nonzero, so $\gamma > 0$ exactly when the game is nondegenerate.
\end{definition}

\begin{theorem}[Margin robustness]
\label{thm:margin}
Let $(\Pi_1, \Pi_2)$ be a nondegenerate rank-one supermodular game with margin $\gamma > 0$. Let $(\tilde{\Pi}_1, \tilde{\Pi}_2)$ be \emph{any} game with $\|\tilde{\Pi}_i - \Pi_i\|_\infty < \gamma/2$ for $i = 1, 2$, where $\|\cdot\|_\infty$ is the entrywise maximum norm. Then $(\tilde{\Pi}_1, \tilde{\Pi}_2)$ admits no better response cycle. In particular, $(\tilde{\Pi}_1, \tilde{\Pi}_2)$ need not be supermodular, low-rank, or otherwise structured.
\end{theorem}

\begin{proof}
Every single-deviation payoff difference of the perturbed game differs from the corresponding difference of the rank-one game by less than $2 \cdot \gamma/2 = \gamma$. Since every single-deviation difference of the rank-one game has absolute value at least $\gamma$, the perturbed difference has the same strict sign. A BRC of $(\tilde{\Pi}_1, \tilde{\Pi}_2)$ is a sequence of single deviations, each strictly improving under $\tilde{\Pi}$; by sign preservation each is strictly improving under $\Pi$, so the same sequence is a BRC of $(\Pi_1, \Pi_2)$. This contradicts Theorem~\ref{thm:rank1}.
\end{proof}

\begin{corollary}[Spectral version]
\label{cor:spectral}
Let $(\tilde{\Pi}_1, \tilde{\Pi}_2)$ be a game whose best rank-one approximations $\Pi_i^{(1)} := \sigma_1^{(i)} \mathbf{u}_1^{(i)} (\mathbf{v}_1^{(i)})^\top$ form a nondegenerate supermodular rank-one game with margin $\gamma$. If the spectral tails satisfy
\[
\Bigl\| \tilde{\Pi}_i - \Pi_i^{(1)} \Bigr\|_\infty \;\leq\; \sum_{\ell \geq 2} \sigma_\ell^{(i)} \, \|\mathbf{u}_\ell^{(i)}\|_\infty \|\mathbf{v}_\ell^{(i)}\|_\infty \;<\; \frac{\gamma}{2} \quad \text{for } i = 1,2,
\]
then $(\tilde{\Pi}_1, \tilde{\Pi}_2)$ admits no better response cycle. Loosely, games whose payoff variation is dominated by a single factor inherit the rank-one conclusion, with the tolerance governed by the margin of the leading factor.
\end{corollary}

\begin{corollary}[Estimation version]
\label{cor:estimation_margin}
Let $y \in \R^{2n}$ encode a nondegenerate rank-one supermodular game with margin $\gamma$. Then every $y'$ with $\|y' - y\|_\infty < \gamma/2$ encodes a game with no better response cycle.
\end{corollary}

Corollary~\ref{cor:estimation_margin} replaces the informal continuity appeal in the multi-market analysis of Section~\ref{sec:rpca}. It gives an explicit sup-norm neighborhood of the latent payoff on which the BRC-free property is guaranteed, which is exactly the form of control delivered by the estimation bounds there. Extending the structural sign-quadrant argument itself to $r \geq 2$, beyond the perturbative reach of Theorem~\ref{thm:margin}, remains open. Example~\ref{ex:rank2} shows that any such extension must involve restrictions beyond rank and supermodularity alone.

%----------------------------------------------------------------------
\section{Manufacturing rank one: multi-market RPCA and the intermediate regime}
\label{sec:rpca}
%----------------------------------------------------------------------

The results of Sections~\ref{sec:rank1}--\ref{sec:rankr} are statements about the latent payoff structure, independent of the information structure. In this section, we embed them in a global game with an explicit RPCA information-processing stage and resolve a key technical obstacle: standard RPCA guarantees require matrix dimensions to be large, but game-theoretic action spaces are finite and fixed.

The resolution is to let the large dimension come not from the action space but from the number of \emph{markets}. Players compete simultaneously in $M$ independent markets governed by the same latent payoff structure, with full discounting ($\delta = 0$) eliminating dynamic strategic considerations. RPCA on the stacked cross-market observations provides the statistical leverage that a single small game cannot. Figure~\ref{fig:stacking} previews the construction.

\subsection{The multi-market global game}
\label{subsec:multimarket}

Fix the base game: two players with action sets $A_i = \{1,\ldots,N_i\}$, $i = 1,2$, and set $n := N_1 N_2$. There is a common latent payoff state $y \in \R^{2n}$ encoding both players' payoffs at all action profiles, exactly as in \citet{veiel2025limits}.

Players compete simultaneously in $M$ independent markets, indexed by $m = 1,\ldots,M$. In each market the game is identical, and full discounting implies no intertemporal strategic considerations: each market is a static game. The markets are linked only through the information structure.

\begin{definition}[Multi-market signal]
\label{def:signal}
In market $m$, player~$i$ privately observes:
\begin{equation}
\label{eq:signal}
x_i^m = y + \sigma \bigl( e_i^m + w_i^m \bigr),
\end{equation}
where:
\begin{enumerate}[label=(\roman*),nosep]
\item $\sigma > 0$ is the noise scale (the global games parameter);
\item $e_i^m \in \R^{2n}$ is a $k$-\emph{sparse} corruption: $|\operatorname{supp}(e_i^m)| \leq k$, with support drawn independently across markets;
\item $w_i^m \in \R^{2n}$ is a dense noise vector with i.i.d.\ entries of mean zero and variance $\sigma_w^2$.
\end{enumerate}
The noise terms $(e_i^m, w_i^m)_{m,i}$ are mutually independent and independent of $y$.
\end{definition}

The sparse corruption $e_i^m$ represents gross misperceptions about specific action-pair payoffs in market~$m$: misreading a contract clause, say, or misestimating a competitor's cost at a particular output level. The dense noise $w_i^m$ represents small measurement errors affecting all entries. The key assumption is that the same latent payoff $y$ governs all markets. The competitive environment has common structure, but observations are market-specific.

\begin{assumption}[Noise regularity]
\label{ass:noise}
Throughout Sections~\ref{sec:rpca} and~\ref{sec:matcomp}, the entries of $w_i^m$ are i.i.d., symmetric about zero, sub-Gaussian with variance proxy $\sigma_w^2$, and the standardized entry $w / \sigma_w$ admits a density bounded below by a constant $c_w > 0$ on a neighborhood of zero. We write $\varepsilon := k/(2n) < 1/2$ for the per-entry corruption probability and $\delta := 1/2 - \varepsilon > 0$ for the corruption margin. The values of the sparse corruptions $e_i^m$ on their support are arbitrary, possibly adversarial.
\end{assumption}

Symmetry guarantees that the population median of each noisy entry equals the true value; without it the median estimator below would be inconsistent. The density lower bound delivers the parametric concentration rate of the sample median. Adversarial corruption values are permitted, which is what makes the contamination bias in Proposition~\ref{prop:intermediate} unavoidable.

The timing within each period is as follows. Nature draws $y$ (once) and the noise $(e_i^m, w_i^m)_{m,i}$. Each player $i$ observes signals $(x_i^1, \ldots, x_i^M)$ across all $M$ markets. Player $i$ processes the signals to form the median signal $\tilde{y}_i$ (and, optionally for estimation, the projected estimate $\hat{y}_i$). Players simultaneously choose actions in each market, conditioning on $\tilde{y}_i$. Since all markets share the same $y$ and the game is static, the final step reduces to $M$ independent copies of the same Bayesian game with beliefs centered on $\tilde{y}_i$.

\subsection{The stacking construction}
\label{subsec:stacking}

Player $i$ stacks her observations column-wise into a matrix:
\begin{equation}
\label{eq:stacking}
X_i := \bigl[ x_i^1 \;\big|\; \cdots \;\big|\; x_i^M \bigr] \;\in\; \R^{2n \times M}.
\end{equation}
Decomposing by~\eqref{eq:signal}:
\begin{equation}
\label{eq:decomp}
X_i = \underbrace{y \cdot \mathbf{1}_M^\top}_{L_0} \;+\; \sigma \underbrace{\bigl[ e_i^1 \;\big|\; \cdots \;\big|\; e_i^M \bigr]}_{S_0} \;+\; \sigma \underbrace{\bigl[ w_i^1 \;\big|\; \cdots \;\big|\; w_i^M \bigr]}_{W_0}.
\end{equation}

\begin{figure}[t]
\centering
\includegraphics[width=\textwidth]{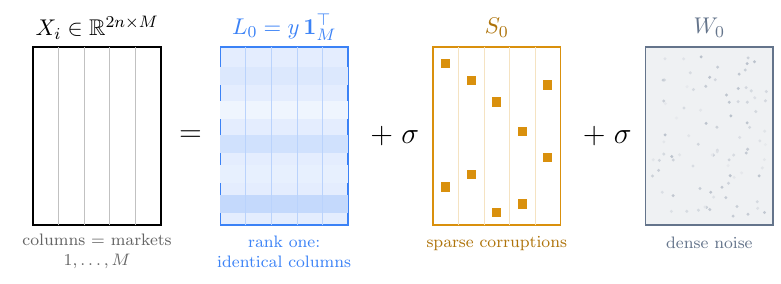}
\caption{\textbf{The stacking construction.} A schematic of equation~\eqref{eq:decomp}. Player~$i$ collects her $M$ market signals column by column into $X_i \in \R^{2n \times M}$. Because every market shares the same latent payoff $y$, the signal component $L_0 = y\,\mathbf{1}_M^\top$ has identical columns and rank exactly one, regardless of the rank of the payoff matrices encoded in $y$ (Lemma~\ref{lem:stacked_structure}). Replication across markets manufactures the low-rank structure out of strategic repetition rather than payoff simplicity, and the resulting low-rank-plus-sparse object is exactly what Robust PCA recovers. The large dimension needed by RPCA comes from $M$, which can grow without changing the action space.}
\label{fig:stacking}
\end{figure}

The key structural observation is that the stacked matrix admits a natural low-rank plus sparse decomposition:

\begin{lemma}[Rank and sparsity of the stacked matrix]
\label{lem:stacked_structure}
In the decomposition~\eqref{eq:decomp}:
\begin{enumerate}[label=(\alph*),nosep]
\item The low-rank component $L_0 = y \cdot \mathbf{1}_M^\top$ has rank exactly $1$, regardless of the rank of $y$ when reshaped as payoff matrices.
\item The sparse component $\sigma S_0$ has at most $k$ nonzero entries per column, hence at most $kM$ nonzero entries total.
\item The dense component $\sigma W_0$ has $\|\sigma W_0\|_F^2 = \sigma^2 \sum_{m=1}^M \|w_i^m\|^2$.
\end{enumerate}
\end{lemma}

\begin{proof}
For (a): $L_0 = y \cdot \mathbf{1}_M^\top$ is the outer product of two nonzero vectors, hence has rank~1. Parts (b) and (c) follow directly from the definitions.
\end{proof}

This is the central insight of the multi-market construction. No matter how complex the underlying payoff structure, the stacked observation matrix always has a rank-one low-rank component. The large dimension needed for RPCA comes from $M$, the number of markets, which can be made arbitrarily large without changing the game-theoretic action space.

\subsection{RPCA recovery}
\label{subsec:recovery}

We consider two regimes, depending on whether dense noise is present.

\subsubsection{Regime 1: purely sparse corruptions
  (\texorpdfstring{$\sigma_w = 0$}{sigma\_w = 0})}

When $w_i^m = 0$ for all $m$, the stacked matrix admits an exact low-rank plus sparse decomposition:
\begin{equation}
\label{eq:puresparse}
X_i = L_0 + \sigma S_0.
\end{equation}
Player~$i$ solves the Principal Component Pursuit (PCP) convex program \citep{candes2011robust}:
\begin{equation}
\label{eq:pcp}
(\hat{L}, \hat{S}) = \arg\min_{L, S} \; \|L\|_* + \lambda \|S\|_1 \quad \text{subject to} \quad L + S = X_i,
\end{equation}
with $\lambda = 1/\sqrt{\max(2n, M)}$. The nuclear norm is the convex envelope of the rank and the $\ell_1$ norm that of the sparsity count, so PCP is the natural convex relaxation of the decomposition problem; it is solved at scale by augmented Lagrangian methods that alternate singular value thresholding on $L$ with entrywise soft thresholding on $S$ \citep{candes2011robust}. Figure~\ref{fig:pcp} illustrates the decomposition the program performs on the stacked matrix.

\begin{figure}[t]
\centering
\includegraphics[width=0.68\textwidth]{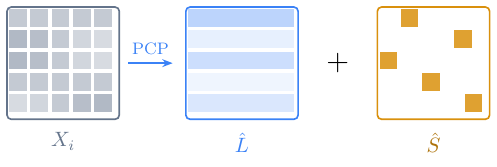}
\caption{\textbf{Principal Component Pursuit on the stacked observation matrix.} The purely sparse regime~\eqref{eq:puresparse}. The convex program~\eqref{eq:pcp} splits $X_i$ into a low-rank estimate $\hat{L}$ and a sparse estimate $\hat{S}$. For the multi-market matrix the low-rank component has identical columns (horizontal bands), reflecting $L_0 = y\,\mathbf{1}_M^\top$, so any single column of $\hat{L}$ returns the latent payoff $y$. Under the incoherence conditions of Lemma~\ref{lem:incoherence} the recovery is exact with high probability (Proposition~\ref{prop:exact_recovery}).}
\label{fig:pcp}
\end{figure}

To apply the PCP recovery guarantee, we must verify the incoherence conditions on $L_0$.

\begin{lemma}[Incoherence of the stacked low-rank component]
\label{lem:incoherence}
The rank-one matrix $L_0 = y \cdot \mathbf{1}_M^\top$ has the (compact) singular value decomposition $L_0 = \|y\| \sqrt{M} \cdot \hat{u} \cdot \hat{v}^\top$, where $\hat{u} = y / \|y\| \in \R^{2n}$ and $\hat{v} = \mathbf{1}_M / \sqrt{M} \in \R^M$. The incoherence parameters are:
\begin{align}
\mu_{\mathrm{left}}(L_0) &:= \frac{2n}{1} \cdot \max_{l \in \{1,\ldots,2n\}} |\hat{u}_l|^2 = 2n \cdot \frac{\max_l |y_l|^2}{\|y\|^2}, \label{eq:mu_left} \\
\mu_{\mathrm{right}}(L_0) &:= \frac{M}{1} \cdot \max_{m \in \{1,\ldots,M\}} |\hat{v}_m|^2 = M \cdot \frac{1}{M} = 1. \label{eq:mu_right}
\end{align}
In particular, $\mu_{\mathrm{right}} = 1$ (perfectly incoherent) independently of $M$.
\end{lemma}

\begin{proof}
Direct computation from the SVD. The right singular vector $\hat{v} = \mathbf{1}_M / \sqrt{M}$ has all entries equal to $1/\sqrt{M}$, giving $\max_m |\hat{v}_m|^2 = 1/M$, and hence $\mu_{\mathrm{right}} = 1$.
\end{proof}

The joint incoherence condition of \citet{candes2011robust} is also satisfied at the same level: $\|\hat{u}\hat{v}^\top\|_\infty = \max_l |y_l| / (\|y\|\sqrt{M})$, which is bounded by $\sqrt{\mu_{\mathrm{left}}/(2nM)}$ by definition of $\mu_{\mathrm{left}}$. The left incoherence $\mu_{\mathrm{left}}$ depends on the payoff structure. It is bounded above by $2n$ (the trivial bound), with equality when all payoff variation is concentrated in a single entry. For ``spread-out'' payoffs, where no single action-pair payoff dominates, $\mu_{\mathrm{left}} = O(1)$.

\begin{proposition}[Exact recovery in the multi-market game]
\label{prop:exact_recovery}
Suppose $\sigma_w = 0$ (purely sparse corruptions), each entry of each $e_i^m$ is corrupted independently with probability $\varepsilon = k/(2n)$ (Bernoulli corruption), and the incoherence $\mu_{\mathrm{left}}$ is bounded by a constant $\mu_0$. Write $n_{(1)} := \max(2n, M)$ and $n_{(2)} := \min(2n, M)$. There exist numerical constants $\rho_r, \rho_s, c' > 0$ such that if
\begin{equation}
\label{eq:pcp_conditions}
\mu_0 \cdot \log^2\bigl(n_{(1)}\bigr) \;\leq\; \rho_r \cdot n_{(2)} \qquad \text{and} \qquad \varepsilon \;\leq\; \rho_s,
\end{equation}
then PCP~\eqref{eq:pcp} with $\lambda = 1/\sqrt{n_{(1)}}$ recovers $L_0$ and $\sigma S_0$ exactly with probability at least $1 - c' \, n_{(1)}^{-10}$. In particular, player~$i$ recovers $y$ exactly from any column of $\hat{L}$.
\end{proposition}

\begin{proof}
We apply Theorem~1.1 of \citet{candes2011robust} to the $2n \times M$ matrix $X_i = L_0 + \sigma S_0$, in its Bernoulli-support form (their Section~2.2 shows the uniform and Bernoulli support models are equivalent for this purpose; the corruption values may be arbitrary). The rank condition $\operatorname{rank}(L_0) \leq \rho_r\, n_{(2)}\, \mu_0^{-1} \log^{-2}(n_{(1)})$ reduces, with $\operatorname{rank}(L_0) = 1$, to the first inequality in~\eqref{eq:pcp_conditions}. The sparsity condition $|\operatorname{supp}(\sigma S_0)| \leq \rho_s \cdot 2nM$ holds in expectation and with high probability under Bernoulli corruption with $\varepsilon \leq \rho_s$ (up to an adjustment of constants). Both incoherence conditions hold by Lemma~\ref{lem:incoherence}. Theorem~1.1 then yields exact recovery with the stated probability.
\end{proof}

\begin{remark}[The admissible range of $M$]
\label{rem:M_range}
Condition~\eqref{eq:pcp_conditions} should be read in two regimes. For $M \leq 2n$ it requires $M \gtrsim \mu_0 \log^2(2n)$: enough markets for the convex program to see the replication. For $M \geq 2n$ it instead requires $\mu_0 \log^2(M) \lesssim n$, that is, $M \leq \exp\bigl(c\sqrt{n/\mu_0}\bigr)$: the PCP guarantee does \emph{not} survive $M \to \infty$ with the action space held fixed, because the matrix becomes extremely thin and the logarithmic factor in the rank condition is driven by the long dimension. The admissible window is wide (super-polynomial in $n$) but bounded. The median-based estimator of the next subsection has no such ceiling, which is one reason we adopt it as the primary information-processing stage.
\end{remark}

However, exact recovery creates a problem for equilibrium selection. It returns the game to complete information, where multiple equilibria can coexist and the global game selection mechanism is inoperative. The intermediate regime requires residual noise that is small but nonzero.

\subsubsection{Regime 2: sparse corruptions plus dense noise
  (\texorpdfstring{$\sigma_w > 0$}{sigma\_w > 0})}

With dense noise, exact recovery is impossible, but the multi-market structure enables estimation with error that vanishes at a controlled rate. We use a two-step procedure that leverages the cross-market replication.

\begin{definition}[RPCA-Averaging estimator]
\label{def:estimator}
The estimator maps player~$i$'s observations $(x_i^1, \ldots, x_i^M)$ into $\hat{y}_i$ in two steps:
\begin{enumerate}[nosep]
\item \textbf{Sparse removal via entry-wise median.} For each entry $l \in \{1,\ldots,2n\}$, compute:
\begin{equation}
\label{eq:median}
\tilde{y}_l := \operatorname{median}_{m \in \{1,\ldots,M\}} \bigl( (x_i^m)_l \bigr).
\end{equation}
\item \textbf{Low-rank projection.} Reshape $\tilde{y} \in \R^{2n}$ into the pair of $N_1 \times N_2$ payoff matrices $(\tilde{\Pi}_1, \tilde{\Pi}_2)$. For each player~$j$, compute the truncated rank-$r$ SVD:
\begin{equation}
\label{eq:svd_proj}
\hat{\Pi}_j := \sum_{\ell=1}^r \hat{\sigma}_\ell \hat{\mathbf{u}}_\ell \hat{\mathbf{v}}_\ell^\top,
\end{equation}
where $\hat{\sigma}_1 \geq \cdots \geq \hat{\sigma}_r$ are the top $r$ singular values of $\tilde{\Pi}_j$. Define $\hat{y}_i := \operatorname{vec}(\hat{\Pi}_1, \hat{\Pi}_2) \in \R^{2n}$.
\end{enumerate}
\end{definition}

\begin{proposition}[Intermediate regime via multi-market averaging]
\label{prop:intermediate}
Suppose each entry of each $e_i^m$ is corrupted independently with probability $\varepsilon = k/(2n) < 1/2$, and let Assumption~\ref{ass:noise} hold with corruption margin $\delta = 1/2 - \varepsilon$. Let $\sigma_w > 0$ be fixed. Then for the RPCA-Averaging estimator of Definition~\ref{def:estimator}:

\begin{enumerate}[label=(\alph*)]
\item \textbf{Sparse removal:} With probability at least $1 - 2n \cdot e^{-2\delta^2 M}$, for every entry $l$ the fraction of corrupted markets is below $1/2$, so the entry-wise median falls weakly between two uncorrupted observations.

\item \textbf{Dense residual:} Conditional on (a), the median estimator satisfies
\begin{equation}
\label{eq:median_bound}
\|\tilde{y} - y\|_\infty \;\leq\; C \cdot \sigma \sigma_w \left( \varepsilon + \sqrt{\frac{\log(2n)}{M}} \right) \qquad \text{with high probability},
\end{equation}
where $C$ depends only on the density lower bound $c_w$ of Assumption~\ref{ass:noise}; hence $\|\tilde{y} - y\|_2 \leq C \sigma \sigma_w \sqrt{2n} \, (\varepsilon + \sqrt{\log(2n)/M})$. The additive term $C \sigma \sigma_w \varepsilon$ is a contamination bias: with adversarial corruption values it cannot be removed by any estimator, and it vanishes if the corruption values are conditionally symmetric about the truth.

\item \textbf{Low-rank projection:} If $y$ corresponds to rank-$r$ payoff matrices, the projection step~\eqref{eq:svd_proj} satisfies $\|\hat{y}_i - y\|_2 \leq 2 \|\tilde{y} - y\|_2$.

\item \textbf{Intermediate regime:} The residual $Z_i := \hat{y}_i - y$ satisfies:
\begin{enumerate}[label=(\roman*),nosep]
\item $\|Z_i\|_2 \leq C' \cdot \sigma \sigma_w \sqrt{2n}\,\bigl( \varepsilon + \sqrt{\log(2n)/M} \bigr)$, which vanishes as $\sigma \to 0$ with $(M, \sigma_w)$ fixed;
\item $\|Z_i\|_2 > 0$ with probability~$1$ for any $\sigma > 0$, $\sigma_w > 0$, $M < \infty$.
\end{enumerate}
\end{enumerate}
\end{proposition}

\begin{proof}
\textbf{Part (a):} For a fixed entry $l$, the number of corrupted markets is $\mathrm{Binomial}(M, \varepsilon)$. By Hoeffding's inequality, the probability that the corrupted fraction reaches $1/2 = \varepsilon + \delta$ is at most $e^{-2\delta^2 M}$. A union bound over $2n$ entries gives the claim. When fewer than half the values are corrupted, the sample median is sandwiched between order statistics of the uncorrupted subsample.

\textbf{Part (b):} Fix entry $l$ and condition on the event of (a), refined so that the corrupted fraction is at most $\varepsilon_M := \varepsilon + \sqrt{\log(4n M)/(2M)} < 1/2$ for every entry (again Hoeffding plus a union bound). Let $F$ denote the distribution of an uncorrupted observation of entry $l$, which is symmetric about $y_l$ with scale $\sigma \sigma_w$. A standard contamination argument shows the median of the contaminated sample lies between the empirical $q^-$ and $q^+$ quantiles of the clean subsample, with $q^\pm = 1/2 \pm \varepsilon_M/(1 - \varepsilon_M)$. Empirical quantiles of the clean subsample concentrate around the corresponding population quantiles at rate $\sqrt{\log(2n)/M}$ (sub-Gaussian tails plus a union bound over entries), and the density lower bound of Assumption~\ref{ass:noise} converts quantile levels to locations: $|F^{-1}(1/2 \pm u) - y_l| \leq u \, \sigma \sigma_w / c_w$. Combining the two sources of deviation gives~\eqref{eq:median_bound}. If corruption values are conditionally symmetric about $y_l$, the contaminated population median is $y_l$ itself, and the $\varepsilon$ term disappears.

\textbf{Part (c):} The set $\mathcal{R} := \{z : \text{each player's reshaped matrix has rank} \leq r\}$ is closed but not convex, so metric projection onto it (which the truncated SVD implements, by Eckart--Young) is not non-expansive in general. The factor two follows from the triangle inequality: writing $P$ for the projection, $\|P(\tilde{y}) - y\| \leq \|P(\tilde{y}) - \tilde{y}\| + \|\tilde{y} - y\| \leq 2\|\tilde{y} - y\|$, where the second step uses that $y \in \mathcal{R}$ is a feasible point, so the projection distance $\|P(\tilde{y}) - \tilde{y}\|$ is at most $\|y - \tilde{y}\|$.

\textbf{Part (d):} The upper bound follows from (b) and (c). Strict positivity holds because $\tilde{y}$ has an absolutely continuous distribution (the sample median of i.i.d.\ draws with a density itself has a density), and the preimage of the single point $y$ under the projection is contained in a lower-dimensional set, hence has probability zero.
\end{proof}

\begin{remark}[Large $M$ does not remove the bias]
\label{rem:bias}
As $M \to \infty$ with $\sigma$ fixed, the bound in (d)(i) tends to $C' \sigma \sigma_w \sqrt{2n}\, \varepsilon > 0$: replication across markets averages away the dense noise but cannot identify which observations are corrupted, and with adversarial corruption values the residual bias of order $\varepsilon \sigma \sigma_w$ is information-theoretically irreducible (it is the standard price of Huber contamination). The intermediate regime is unaffected, since the entire error is proportional to $\sigma$ and vanishes in the global-game limit $\sigma \to 0$.
\end{remark}

The two-step estimator separates the roles of robust estimation (median for sparse removal) and structural projection (SVD for low rank). One can alternatively apply Stable PCP \citep{zhou2010stable} to the full stacked matrix directly, but the error bounds for Stable PCP scale as $O(\sqrt{n_1 n_2} \cdot \delta)$, which is too loose when averaged per column: the improvement from $M$ markets is absorbed by the dimensional factor. The two-step procedure avoids this by exploiting the replication structure directly. The conceptual contribution of the RPCA framework nonetheless remains central: it identifies the decomposition $X_i = L_0 + \sigma S_0 + \sigma W_0$ as the natural signal structure, and the median step implements the sparse-removal component of RPCA using cross-market replication as a substitute for convex-optimization-based PCP.

\subsection{Connecting to limit uniqueness}
\label{subsec:connecting}

We now connect the multi-market RPCA estimator to the equilibrium selection analysis.

\begin{definition}[Multi-market global game]
For each $\sigma > 0$ and $M \in \N$, the \emph{multi-market global game} $\Gamma(\sigma, M)$ is the Bayesian game in which:
\begin{enumerate}[label=(\roman*),nosep]
\item The latent state $y \in S \subset \R^{2n}$ is drawn from a prior $\nu_0$ on the unit sphere;
\item Each player $i$ observes $(x_i^1, \ldots, x_i^M)$ as in Definition~\ref{def:signal};
\item Player $i$ forms the effective signal $\tilde{y}_i$ via the median step (Step~1) of the RPCA-Averaging estimator (Definition~\ref{def:estimator});
\item In each market $m$, player $i$ chooses an action $a_i^m \in A_i$ to maximize expected payoff given beliefs about the opponent's play, where payoffs in market $m$ are determined by $y$.
\end{enumerate}
\end{definition}

We deliberately use the unprojected median signal $\tilde{y}_i$, rather than the full estimate $\hat{y}_i$, in the strategic analysis. The low-rank projection of Step~2 improves point estimation, but it collapses the signal onto the rank-$r$ variety, a measure-zero subset of $\R^{2n}$, and thereby destroys the full-dimensional local support of the effective noise on which the contagion argument relies. The projection remains available to either player as a post-processing step for estimation. What matters for the game is the signal on which strategies condition.

Since the game is static (full discounting) and markets share the same $y$, the strategy in each market depends only on the common signal $\tilde{y}_i$. Hence the multi-market game reduces to $M$ independent copies of a single-market game with the effective signal $\tilde{y}_i$.

\begin{definition}[Effective noise]
The \emph{effective noise} in the multi-market game is the random variable
\[
Z_i := \tilde{y}_i - y.
\]
By Proposition~\ref{prop:intermediate}(a)--(b), $\|Z_i\|_\infty = O\bigl(\sigma \sigma_w (\varepsilon + \sqrt{\log n / M})\bigr)$ with high probability and $\|Z_i\|_2 > 0$ a.s. Its finer distributional properties are collected in Proposition~\ref{prop:signal_properties} below.
\end{definition}

The selection analysis now splits cleanly into two parts: a proposition collecting everything we can prove about the effective signals, and a conjecture recording the selection statement itself. The partition follows Remark~\ref{rem:conditional}.

\begin{proposition}[Vanishing, cycle-free, full-support effective signals]
\label{prop:signal_properties}
Let $y \in \calS^*$ be a quasi-concave supermodular payoff such that each player's payoff matrix has rank one and is nondegenerate in the sense of Definition~\ref{def:margin}, with margin $\gamma > 0$. Fix $M$ and $\sigma_w > 0$. Then in the multi-market global game $\Gamma(\sigma, M)$:
\begin{enumerate}[label=(\alph*),nosep]
\item \textbf{Vanishing noise:} $\|Z_i\|_\infty \to 0$ in probability as $\sigma \downarrow 0$;
\item \textbf{Cycle-free signals:} on the event $\|Z_i\|_\infty < \gamma/2$, whose probability tends to one as $\sigma \downarrow 0$, the realized signal $\tilde{y}_i$ encodes a game with no better response cycle;
\item \textbf{Full local support:} conditional on the sparse-removal event of Proposition~\ref{prop:intermediate}(a), the entries of $Z_i$ are independent across $l$ and each admits a density that is strictly positive in a neighborhood of zero, so $Z_i$ has full local support in $\R^{2n}$. If, in addition, the corruption values are conditionally symmetric about the truth (in particular, if $\varepsilon = 0$), each entry is symmetric about zero.
\end{enumerate}
\end{proposition}

\begin{proof}
\emph{(a)} By Proposition~\ref{prop:intermediate}(a)--(b), the error bound carries the factor $\sigma$ throughout and hence vanishes as $\sigma \downarrow 0$ with $(M, \sigma_w)$ fixed.

\emph{(b)} By Theorem~\ref{thm:rank1}, $y$ admits no better response cycles, and by Corollary~\ref{cor:estimation_margin} the same holds for every $y'$ with $\|y' - y\|_\infty < \gamma/2$. Combine with (a).

\emph{(c)} Conditional on sparse removal, $\tilde{y}_l$ depends on the market observations of entry $l$ alone, and these are independent across $l$, giving independence of the entries of $Z_i$. Each $\tilde{y}_l$ is a sample median of $M$ real observations, of which a majority are i.i.d.\ draws from a distribution with a density bounded below near $y_l$ (Assumption~\ref{ass:noise}); the sandwich property of the median then implies $\tilde{y}_l$ itself has a density strictly positive in a neighborhood of $y_l$. When corruption values are conditionally symmetric about $y_l$, the contaminated entry distribution is symmetric about $y_l$, and the sample median of i.i.d.\ draws from a symmetric distribution is symmetric about its center.
\end{proof}

\begin{conjecture}[Multi-market limit uniqueness for rank-one games]
\label{conj:multimarket_lu}
Under the hypotheses of Proposition~\ref{prop:signal_properties}, limit uniqueness holds in $\Gamma(\sigma, M)$ as $\sigma \to 0$: for any open $O \ni y$ contained in~$\calS^*$,
\[
\lim_{\sigma \downarrow 0} \Pr\bigl( \exists\, i \text{ s.t.\ } |\mathrm{ICR}_i^{\sigma,M}(\tilde{y}_i)| > 1 \text{ and } \tilde{y}_i \in O \bigr) = 0,
\]
where $\mathrm{ICR}_i^{\sigma,M}(\tilde{y}_i)$ denotes player $i$'s set of interim correlated rationalizable actions given the signal $\tilde{y}_i$.
\end{conjecture}

The supporting argument is the contagion mechanism of \citet{veiel2025limits}. Parts (a) and (b) of Proposition~\ref{prop:signal_properties} deliver signals that concentrate on a cycle-free neighborhood of $y$, and part (c) delivers dominance regions of positive probability in signal space for every $\sigma > 0$. With no risk-dominant BRC available to obstruct it on $O$, iterated deletion from the dominance regions should then propagate inward and pin down a unique rationalizable action, exactly as in his Theorem~4.1.

\begin{remark}[Why this is a conjecture and not a theorem]
\label{rem:noise_class}
Theorem~4.1 of \citet{veiel2025limits} is established for a specified class of symmetric multi-dimensional noise. The effective noise $Z_i$ generated by the median is entrywise independent with full local support (Proposition~\ref{prop:signal_properties}(c)), but at finite $M$ it is not a member of that class: the median is a nonlinear functional of the raw noise, its distribution depends on $M$, and with adversarial corruption values exact entrywise symmetry can fail. Completing the argument therefore requires extending Veiel's characterization to entrywise-independent noise with full local support. The absence of BRCs on the relevant neighborhood, which Proposition~\ref{prop:signal_properties}(b) delivers, is exactly the substantive condition his characterization isolates, so we expect the extension to hold; but we have not proved it, and we prefer a clearly labeled conjecture to a gap inside a proposition. Every estimation result in this section and the next is unaffected and unconditional.
\end{remark}

The number of markets $M$ plays a dual role. Statistically, $M$ determines the quality of RPCA recovery, with estimation error $O\bigl(\sigma \sigma_w (\varepsilon + \sqrt{n/M})\bigr)$. Larger $M$ gives better denoising, but the intermediate regime is preserved for any finite~$M$. Strategically, since markets are independent and the game is static, $M$ does not affect the strategic analysis of any individual market. The entire effect of $M$ operates through the information channel. This separation is key: the large dimension needed for RPCA comes from the number of markets, not from the complexity of the strategic interaction.

For the multi-market setting, the extended logical chain is:
\begin{multline*}
\text{Multi-market + RPCA} \;\xRightarrow{\text{Prop.~\ref{prop:intermediate}}}\; \text{Intermediate regime} \\
\;\xRightarrow{\text{Prop.~\ref{prop:signal_properties}}}\; \text{Cycle-free, full-support signals}
\;\xRightarrow{\text{Conj.~\ref{conj:multimarket_lu}}}\; \text{Limit uniqueness.}
\end{multline*}
Every link except the last is a theorem; the last is the noise-class extension of Remark~\ref{rem:noise_class}.

\subsection{Beyond rank one: multi-market RPCA as a sparsity channel}
\label{subsec:beyond_rank1}

For games with rank $r \geq 2$, where BRCs may exist at the latent payoff $y$, the multi-market structure provides an additional channel for limit uniqueness through the sparse corruption structure.

\begin{proposition}[Sparse cycle-breaking with multi-market data]
\label{prop:sparse_breaking}
Suppose $y$ supports a BRC of length $m$ involving action subsets $B_1 \times B_2$. After RPCA-Averaging (Definition~\ref{def:estimator}), the estimation residual $\hat{Z}_i := \hat{y}_i - y$ has the property that for each entry $l \in \{1,\ldots,2n\}$, $(\hat{Z}_i)_l$ is determined by the dense noise alone: the sparse corruption has been eliminated (with high probability).

Consequently, the entries of $\hat{Z}_i$ corresponding to distinct action-pair payoffs are \emph{approximately independent} (inheriting the independence of $w_i^m$ across entries), unlike the correlated structure that would arise from RPCA on a single observation.

In particular, if the dense noise $w_i^m$ is symmetric (invariant under sign changes and permutations of action labels), the residual $\hat{Z}_i$ inherits this symmetry \emph{approximately} as $M \to \infty$, recovering the symmetry conditions required by Veiel's multilinear form analysis (Proposition~5.2 therein).
\end{proposition}

\begin{proof}[Proof sketch]
Conditional on the event in Proposition~\ref{prop:intermediate}(a) (sparse removal succeeds), the median $\tilde{y}_l$ for each entry $l$ depends only on the dense noise values $\{w_{i,l}^m\}_{m=1}^M$. Since these are i.i.d.\ across entries $l$ (by assumption on $w_i^m$), the entries of $\tilde{y} - y$ are approximately independent. The low-rank projection (step~2) introduces dependence, but only within the rank-$r$ subspace, preserving the approximate symmetry of the noise in the orthogonal complement.
\end{proof}

This proposition suggests that multi-market RPCA, by stripping out the sparse asymmetric component of the noise, can restore the approximate symmetry conditions under which Veiel's characterization applies. Formalizing this, and in particular quantifying the rate at which symmetry is restored as $M \to \infty$, is an important direction for future work.

If both $M \to \infty$ and $\sigma \to 0$ simultaneously, the relative rates matter, with one caveat from Remark~\ref{rem:bias}: the dense-noise component of the residual shrinks like $\sigma\sqrt{n/M}$, but the contamination bias of order $\varepsilon \sigma \sigma_w$ shrinks only with $\sigma$. If $M$ is fixed and $\sigma \to 0$, the residual is $O(\sigma)$ and the standard global game analysis applies with effective noise of scale $\sigma \sigma_w (\varepsilon + \sqrt{n/M})$. The natural global game limit is $\sigma \to 0$ with $(M, \sigma_w)$ fixed, in which case the multi-market RPCA framework provides a well-defined intermediate regime.

%----------------------------------------------------------------------
\section{Partial observation and matrix completion}
\label{sec:matcomp}
%----------------------------------------------------------------------

The multi-market framework of Section~\ref{sec:rpca} assumes that each player observes a complete signal vector $x_i^m \in \R^{2n}$ in every market~$m$. In practice, players may observe only a subset of action-pair payoffs in each market, because they have experimented with only some action combinations or because entire markets are unobserved. This section extends the RPCA-Averaging estimator to partial observations, drawing on the matrix completion literature. We consider three observation models in increasing order of economic interest: missing columns (entire markets unobserved), missing entries with uniform random observation, and structured missingness arising from equilibrium play.

Throughout, we maintain the notation of Section~\ref{sec:rpca}: the stacked matrix is $X_i = L_0 + \sigma S_0 + \sigma W_0 \in \R^{2n \times M}$, with $L_0 = y \cdot \mathbf{1}_M^\top$ of rank one, $S_0$ columnwise $k$-sparse, and $W_0$ dense noise with i.i.d.\ entries of variance $\sigma_w^2$.

%----------------------------------------------------------------------
\subsection{Case 1: missing columns}
\label{subsec:missing_columns}
%----------------------------------------------------------------------

The simplest partial observation model assumes that player~$i$ observes complete signal vectors in only a subset of markets.

\begin{definition}[Column observation model]
\label{def:column_obs}
Player~$i$ observes signals $x_i^m$ for markets $m \in \Omega_i \subseteq \{1,\ldots,M\}$, where $|\Omega_i| = M_{\mathrm{obs}}$. No information is available from markets $m \notin \Omega_i$.
\end{definition}

This is equivalent to restricting the stacked matrix to its observed columns: $X_i^{\Omega} := [x_i^m]_{m \in \Omega_i} \in \R^{2n \times M_{\mathrm{obs}}}$. The restricted matrix decomposes as
\[
X_i^{\Omega} = y \cdot \mathbf{1}_{M_{\mathrm{obs}}}^\top + \sigma [e_i^m]_{m \in \Omega_i} + \sigma [w_i^m]_{m \in \Omega_i}.
\]
The low-rank component $y \cdot \mathbf{1}_{M_{\mathrm{obs}}}^\top$ is still rank one, and the RPCA-Averaging estimator (Definition~\ref{def:estimator}) applies verbatim with $M$ replaced by $M_{\mathrm{obs}}$.

\begin{proposition}[Intermediate regime with missing columns]
\label{prop:missing_columns}
Under the column observation model (Definition~\ref{def:column_obs}), the RPCA-Averaging estimator applied to $X_i^{\Omega}$ satisfies the conclusions of Proposition~\ref{prop:intermediate} with $M$ replaced by $M_{\mathrm{obs}}$ throughout. In particular:
\begin{enumerate}[label=(\alph*),nosep]
\item Sparse removal succeeds with probability at least $1 - 2n \cdot e^{-2\delta^2 M_{\mathrm{obs}}}$;
\item The estimation error satisfies $\|\hat{y}_i - y\|_2 \leq C' \cdot \sigma \sigma_w \sqrt{2n}\,\bigl(\varepsilon + \sqrt{\log(2n) / M_{\mathrm{obs}}}\bigr)$;
\item The intermediate regime is preserved for any $M_{\mathrm{obs}} < \infty$ and $\sigma > 0$.
\end{enumerate}
\end{proposition}

\begin{proof}
Since every column of $L_0$ is identical, the subset of observed columns carries the same per-entry information as the full matrix. The entry-wise median is computed over $M_{\mathrm{obs}}$ markets instead of $M$; the Hoeffding bound in Proposition~\ref{prop:intermediate}(a) and the quantile concentration in Proposition~\ref{prop:intermediate}(b), including its contamination bias term, carry through with $M_{\mathrm{obs}}$ in place of $M$. The SVD projection step is unchanged.
\end{proof}

Fewer observed markets degrade the estimation rate by a factor of $\sqrt{M/M_{\mathrm{obs}}}$, but the qualitative structure is entirely preserved: intermediate regime, sparse removal, low-rank projection. The only requirement is that $M_{\mathrm{obs}}$ is large enough for the median to work; $M_{\mathrm{obs}} \gtrsim \log(2n)$ suffices for the Chernoff bound in part~(a).

%----------------------------------------------------------------------
\subsection{Case 2: missing entries with uniform observation}
\label{subsec:missing_entries}
%----------------------------------------------------------------------

The economically more interesting case is when player~$i$ observes only a subset of action-pair payoffs \emph{within} each market. This is the natural matrix completion setting: the player knows payoffs at action pairs she has explored, but not at pairs she has not tried.

\begin{definition}[Entry observation model]
\label{def:entry_obs}
For each market $m$ and entry $l \in \{1,\ldots,2n\}$, player~$i$ observes $(x_i^m)_l$ independently with probability $p \in (0,1]$. Let $\Omega \subseteq \{1,\ldots,2n\} \times \{1,\ldots,M\}$ denote the (random) set of observed positions, and let $\mathcal{P}_\Omega$ denote the entry-wise projection onto $\Omega$:
\[
(\mathcal{P}_\Omega(A))_{l,m} = \begin{cases} A_{l,m} & \text{if } (l,m) \in \Omega, \\ 0 & \text{otherwise.} \end{cases}
\]
The observed data is $\mathcal{P}_\Omega(X_i) = \mathcal{P}_\Omega(L_0 + \sigma S_0 + \sigma W_0)$.
\end{definition}

This combines matrix completion with robust estimation: we must recover a rank-one matrix from partial, sparsely corrupted, noisy observations. As noted in Lemma~\ref{lem:incoherence}, the stacked low-rank component $L_0 = y \cdot \mathbf{1}_M^\top$ has favorable incoherence properties, with $\mu_{\mathrm{right}} = 1$ and $\mu_{\mathrm{left}} = O(1)$ for spread-out payoffs, which are precisely the conditions required for matrix completion.

\subsubsection{Adaptation of the RPCA-Averaging estimator}

The two-step RPCA-Averaging procedure adapts cleanly to partial observations.

\begin{definition}[RPCA-Averaging estimator with missing entries]
\label{def:estimator_mc}
Given partially observed data $\mathcal{P}_\Omega(X_i)$:
\begin{enumerate}[nosep]
\item \textbf{Sparse removal via partial-observation median.} For each entry $l \in \{1,\ldots,2n\}$, let $\Omega_l := \{m \in \{1,\ldots,M\} : (l,m) \in \Omega\}$ denote the set of markets where entry $l$ is observed. Compute:
\begin{equation}
\label{eq:median_mc}
\tilde{y}_l := \operatorname{median}_{m \in \Omega_l} \bigl( (x_i^m)_l \bigr).
\end{equation}
\item \textbf{Low-rank projection.} Identical to Step~2 of Definition~\ref{def:estimator}: reshape $\tilde{y}$ into payoff matrices and compute the truncated rank-$r$ SVD.
\end{enumerate}
\end{definition}

The key change is that the median for entry $l$ is computed over $|\Omega_l|$ observations rather than $M$. Since each entry is observed independently with probability $p$, $|\Omega_l| \sim \mathrm{Binomial}(M, p)$ and concentrates around $pM$.

\begin{proposition}[Intermediate regime with missing entries]
\label{prop:missing_entries}
Under the entry observation model (Definition~\ref{def:entry_obs}), suppose $k < n$ and $pM \geq 8 \log(4n)$. Then for the RPCA-Averaging estimator with missing entries (Definition~\ref{def:estimator_mc}):
\begin{enumerate}[label=(\alph*),nosep]
\item \textbf{Sufficient observations per entry.} With probability at least $1 - 2n \cdot e^{-pM/8}$, every entry $l$ has $|\Omega_l| \geq pM/2$.

\item \textbf{Sparse removal under partial observation.} Conditional on the event in~(a), sparse removal succeeds: for each entry $l$, fewer than half of the observed values are corrupted, with probability at least $1 - 2n \cdot e^{-\delta^2 pM}$.

\item \textbf{Dense residual.} Conditional on (a) and (b), the median estimator satisfies
\begin{equation}
\label{eq:median_bound_mc}
\|\tilde{y} - y\|_\infty \;\leq\; C \cdot \sigma \sigma_w \left( \varepsilon + \sqrt{\frac{\log(2n)}{pM}} \right) \qquad \text{with high probability.}
\end{equation}

\item \textbf{Intermediate regime.} The residual $Z_i := \hat{y}_i - y$ satisfies:
\begin{enumerate}[label=(\roman*),nosep]
\item $\|Z_i\|_2 \leq C' \cdot \sigma \sigma_w \sqrt{2n}\,\bigl(\varepsilon + \sqrt{\log(2n) / (pM)}\bigr)$, which vanishes as $\sigma \to 0$ for fixed $p, M$;
\item $\|Z_i\|_2 > 0$ with probability~$1$ for any $\sigma > 0$, $\sigma_w > 0$, $pM < \infty$.
\end{enumerate}
\end{enumerate}
\end{proposition}

\begin{proof}
\textbf{Part (a):} By a multiplicative Chernoff bound, $\Pr(|\Omega_l| < pM/2) \leq e^{-pM/8}$. A union bound over $2n$ entries gives $\Pr(\exists\, l : |\Omega_l| < pM/2) \leq 2n \cdot e^{-pM/8}$.

\textbf{Part (b):} Conditional on $|\Omega_l| \geq pM/2$, the median over $\Omega_l$ is robust to corruption as follows. The corruption event (entry $l$ is corrupted in market $m$) and the observation event (entry $l$ is observed in market $m$) are independent by assumption. Among the $|\Omega_l|$ observed values of entry $l$, the number that are corrupted is $\mathrm{Binomial}(|\Omega_l|, \varepsilon)$ with $\varepsilon = k/(2n) < 1/2$. By Hoeffding's inequality, the corrupted fraction reaches $1/2 = \varepsilon + \delta$ with probability at most $e^{-2\delta^2 |\Omega_l|} \leq e^{-\delta^2 pM}$. A union bound over $2n$ entries gives the claim.

\textbf{Part (c):} Conditional on the events in (a) and (b), the contaminated median for entry $l$ is sandwiched between empirical quantiles of the clean observed subsample at levels $1/2 \pm \varepsilon'$, with $\varepsilon'$ bounded by a constant multiple of $\varepsilon$, exactly as in the proof of Proposition~\ref{prop:intermediate}(b). Since $|\Omega_l^{\mathrm{clean}}| \geq |\Omega_l|/2 \geq pM/4$, quantile concentration plus the density lower bound of Assumption~\ref{ass:noise} give $|\tilde{y}_l - y_l| \leq C \sigma \sigma_w (\varepsilon + \sqrt{\log(2n)/(pM)})$ with high probability uniformly over $l$, by a union bound.

\textbf{Part (d):} Follows from (c) exactly as in Proposition~\ref{prop:intermediate}(d): the SVD projection step at most doubles the error, by the argument of Proposition~\ref{prop:intermediate}(c), and the noise retains a density.
\end{proof}

The dense-noise component of the estimation rate degrades by a factor of $1/\sqrt{p}$ relative to the full-observation case, while the contamination bias term is unchanged. Importantly, the intermediate regime is preserved for \emph{any} $p > 0$. Partial observation slows convergence but does not alter the qualitative structure. The sample complexity condition $pM \geq 8\log(4n)$ ensures that every entry is observed enough times for the median to be reliable, which is mild when $M$ is moderately large.

\begin{remark}[Connection to matrix completion]
\label{rem:matcomp}
The standard matrix completion literature \citep{candes2010matrix,keshavan2010matrix} recovers a rank-$r$ matrix in $\R^{n_1 \times n_2}$ from $O(\mu_0 \cdot r \cdot \max(n_1, n_2) \cdot \mathrm{polylog})$ observed entries, under incoherence. For the stacked matrix $L_0 \in \R^{2n \times M}$ with rank one and $\mu_{\mathrm{right}} = 1$, this requires $p \cdot 2nM \gtrsim \mu_0 \cdot \max(2n, M) \cdot \mathrm{polylog}$, i.e., $p \gtrsim \mu_0 \cdot \mathrm{polylog} / \min(2n, M)$. The RPCA-Averaging estimator avoids the need to solve a nuclear-norm minimization problem by exploiting the replication structure (every column of $L_0$ is identical), which allows a simpler entry-wise median computation. The median-based approach requires the weaker condition $pM \gtrsim \log(n)$, which does not impose any relationship between $p$ and $n$.
\end{remark}

\subsubsection{Approximate symmetry under partial observation}

A concern is whether missing data disrupts the noise symmetry properties required by Veiel's analysis. With uniform random observation, the answer is negative.

\begin{proposition}[Approximate symmetry with uniform missingness]
\label{prop:symmetry_mc}
Under the entry observation model (Definition~\ref{def:entry_obs}) with uniform observation probability $p$, conditional on the events in Proposition~\ref{prop:missing_entries}(a)--(b), the median residuals $(\tilde{y}_l - y_l)_{l=1}^{2n}$ satisfy:
\begin{enumerate}[label=(\roman*),nosep]
\item The residuals are approximately independent across entries $l$, since $\tilde{y}_l$ depends only on the dense noise values $\{w_{i,l}^m : m \in \Omega_l^{\mathrm{clean}}\}$, which are independent across $l$ by assumption.
\item The marginal distribution of $\tilde{y}_l - y_l$ is the same for all entries $l$ (up to the randomness in $|\Omega_l|$), since each entry faces the same observation probability $p$ and the dense noise is identically distributed across entries.
\end{enumerate}
Consequently, Proposition~\ref{prop:sparse_breaking} extends to the partial observation setting: the effective noise $Z_i$ inherits the approximate symmetry of the dense noise, and the conclusions regarding Veiel's multilinear form analysis carry over.
\end{proposition}

\begin{proof}
Part~(i) follows from the independence of $w_{i,l}^m$ across entries $l$ (maintained from the full-observation case). Part~(ii) follows from the exchangeability of the observation model: $|\Omega_l|$ has the same marginal distribution $\mathrm{Binomial}(M,p)$ for every entry $l$, and conditional on $|\Omega_l|$, the median of $|\Omega_l|$ i.i.d.\ noise draws has the same distribution regardless of $l$. The extension of Proposition~\ref{prop:sparse_breaking} follows because the approximate independence and identical marginals of the residuals are the only properties of $\tilde{y} - y$ used in that proof.
\end{proof}

%----------------------------------------------------------------------
\subsection{Case 3: structured missingness from strategic observation}
\label{subsec:strategic_obs}
%----------------------------------------------------------------------

The most economically substantive extension replaces uniform random missingness with \emph{endogenous} observation patterns. Player~$i$ observes payoffs at action pairs that were actually played in each market, either by herself or by market participants she monitors. This creates structured missingness that depends on the equilibria realized across markets.

\begin{definition}[Strategic observation model]
\label{def:strategic_obs}
In market $m$, an equilibrium $\alpha^m \in A_1 \times A_2$ is played. Player~$i$ observes $(x_i^m)_l$ only for entries $l$ corresponding to the action pair $\alpha^m$ and possibly a neighborhood of $\alpha^m$ in the action space (e.g., payoffs at nearby action pairs from local experimentation). Let $\Omega_i^m \subseteq \{1,\ldots,2n\}$ denote the set of entries observed by player~$i$ in market~$m$.
\end{definition}

Unlike the uniform model, the sets $\Omega_i^m$ are not drawn independently of the payoff structure. They depend on which equilibria arise, which in turn depend on the payoffs. However, the rank-one structure of $L_0$ provides a special property that partially insulates the estimation problem from this endogeneity.

\begin{remark}[Timing and selection]
\label{rem:timing}
The strategic observation model contains a latent circularity if read within a single period: play generates observations, but players need observations to play. We resolve it with a two-epoch interpretation. The observation sets $\Omega_i^m$ are generated by an earlier epoch of play, whose noise realizations are independent of the current epoch's dense noise $w_i^m$; players then estimate $y$ from the inherited observation pattern and play the current epoch's static game. Condition~(c) of Proposition~\ref{prop:strategic_estimation} handles the independence of the sparse corruption from the observation pattern; the two-epoch timing supplies the analogous independence for the dense noise. Without it, conditioning on which entries were observed would distort the noise distribution and introduce selection bias into the median.
\end{remark}

\subsubsection{Why replication helps: column homogeneity}

The key observation is that every column of $L_0 = y \cdot \mathbf{1}_M^\top$ is identical. Observing entry $l$ in \emph{any} market is equally informative about $y_l$. What matters for estimation is not in which market an entry is observed, but whether, across all $M$ markets, each entry $l$ is observed sufficiently often.

\begin{definition}[Coverage]
\label{def:coverage}
The \emph{coverage} of the observation pattern $\{\Omega_i^m\}_{m=1}^M$ is the set $\Omega_i^{\cup} := \bigcup_{m=1}^M \Omega_i^m \subseteq \{1,\ldots,2n\}$. The observation pattern has \emph{full coverage} if $\Omega_i^{\cup} = \{1,\ldots,2n\}$. For each entry $l$, define $M_l := |\{m : l \in \Omega_i^m\}|$ as the number of markets in which entry $l$ is observed.
\end{definition}

\begin{proposition}[Estimation under strategic observation]
\label{prop:strategic_estimation}
Suppose the observation pattern satisfies the following conditions:
\begin{enumerate}[label=(\alph*),nosep]
\item \textbf{Full coverage}: $\Omega_i^{\cup} = \{1,\ldots,2n\}$.
\item \textbf{Minimum depth}: $M_l \geq M_{\min}$ for all $l \in \{1,\ldots,2n\}$, where $M_{\min} \geq 8\log(4n)$.
\item \textbf{Independent corruption}: the support of $e_i^m$ is independent of $\Omega_i^m$ for each $m$.
\end{enumerate}
Then the RPCA-Averaging estimator with missing entries (Definition~\ref{def:estimator_mc}) satisfies:
\begin{equation}
\label{eq:strategic_bound}
\|\hat{y}_i - y\|_\infty \leq C \cdot \sigma \sigma_w \left( \varepsilon + \sqrt{\frac{\log(2n)}{M_{\min}}} \right)
\end{equation}
with probability at least $1 - 4n \cdot e^{-\delta^2 M_{\min}}$.
\end{proposition}

\begin{proof}
For each entry $l$, the median is computed over $M_l \geq M_{\min}$ observed values. The independence of corruption supports from observation sets (condition~(c)) ensures that the fraction of corrupted observations among $\Omega_l$ has the same distribution as in the uniform case, so the sparse removal argument of Proposition~\ref{prop:missing_entries}(b) applies with $pM$ replaced by $M_{\min}$. The quantile concentration argument of Proposition~\ref{prop:intermediate}(b) then gives $|\tilde{y}_l - y_l| \leq C \sigma \sigma_w (\varepsilon + \sqrt{\log(2n)/M_{\min}})$ for each $l$. A union bound over $2n$ entries yields the claim.
\end{proof}

\subsubsection{Equilibrium multiplicity as informational diversity}

The coverage and minimum-depth conditions in Proposition~\ref{prop:strategic_estimation} have a striking economic interpretation. If the same equilibrium is played in every market, the observation set $\Omega_i^m$ is the same across markets, and entries corresponding to off-equilibrium action pairs are never observed. Coverage fails. But if different markets play \emph{different} equilibria, different action pairs are explored across markets, and the union $\Omega_i^{\cup}$ can cover the full action space.

This yields a complementarity between equilibrium multiplicity and information: the existence of multiple equilibria across markets is informationally beneficial because it generates exploration of the payoff space. The multi-market structure helps not just through replication (as in Section~\ref{sec:rpca}) but through \emph{coverage}.

\begin{definition}[Equilibrium diversity]
\label{def:eq_diversity}
Let $\calA^m \subseteq \{1,\ldots,2n\}$ denote the set of payoff entries revealed by the equilibrium played in market~$m$. The \emph{equilibrium diversity} of the multi-market system is the number of distinct equilibria played, $K := |\{\calA^m : m = 1,\ldots,M\}|$. The \emph{revealed set} is $\calA^{\cup} := \bigcup_{m=1}^M \calA^m$.
\end{definition}

\begin{proposition}[Coverage from equilibrium diversity]
\label{prop:diversity_coverage}
Suppose that:
\begin{enumerate}[label=(\roman*),nosep]
\item $K$ distinct equilibria are played across markets, with equilibrium $j$ ($j = 1,\ldots,K$) revealing entries $\calA_j \subseteq \{1,\ldots,2n\}$;
\item $\bigcup_{j=1}^K \calA_j = \{1,\ldots,2n\}$ (the union of revealed sets covers all entries);
\item Each equilibrium $j$ is played in at least $M_j \geq M_{\min}$ markets.
\end{enumerate}
Then the strategic observation model satisfies the conditions of Proposition~\ref{prop:strategic_estimation} with full coverage and $M_l \geq M_{\min}$ for all $l$.
\end{proposition}

\begin{proof}
Full coverage follows from $\bigcup_j \calA_j = \{1,\ldots,2n\}$ and the fact that $\Omega_i^m \supseteq \calA^m$ (the player observes at least the entries revealed by the equilibrium played in each market). For any entry $l$, there exists at least one equilibrium $j$ with $l \in \calA_j$, and this equilibrium is played in at least $M_j \geq M_{\min}$ markets, so $M_l \geq M_{\min}$.
\end{proof}

Condition~(ii) of Proposition~\ref{prop:diversity_coverage} is a \emph{combinatorial coverage condition} rather than a probabilistic one. It requires that the set of equilibria realized across markets, taken together, explore all action pairs. This is a condition on the game's equilibrium structure and on how equilibria are selected across markets, not on the noise or the observation technology. It is satisfied, for instance, whenever the base game has enough equilibria that their supports collectively span the action space, and nature's equilibrium selection mechanism across markets is sufficiently varied.

\subsubsection{Heterogeneous observation depth and the symmetry caveat}

Under strategic observation, the number of observations per entry, $M_l$, typically varies across entries. Action pairs corresponding to frequently played equilibria are observed more often than those corresponding to rare equilibria. This creates \emph{heterogeneous effective noise} in the median estimates.

\begin{remark}[Entry-dependent residual variance]
\label{rem:heterogeneous}
When $M_l$ varies across entries $l$, the median residual $\tilde{y}_l - y_l$ has entry-dependent variance: $\operatorname{Var}(\tilde{y}_l - y_l) \approx \pi \sigma^2 \sigma_w^2 / (2 M_l)$ (the asymptotic variance of the sample median for i.i.d.\ normal data). Frequently observed entries ($M_l$ large) have tighter estimates than rarely observed ones ($M_l$ small).
\end{remark}

This heterogeneity introduces a potential complication for Veiel's multilinear form analysis, which requires approximate symmetry of the noise across entries. The residuals $(\tilde{y}_l - y_l)_l$ are still approximately independent (as in Proposition~\ref{prop:symmetry_mc}(i)), but they are no longer identically distributed.

\begin{proposition}[Conditional symmetry under heterogeneous depth]
\label{prop:conditional_symmetry}
Under the strategic observation model, suppose the observation depths $\{M_l\}_{l=1}^{2n}$ are realized (and hence known to the analyst). Conditional on the depths, the median residuals $\{\tilde{y}_l - y_l\}_{l=1}^{2n}$ are independent, with each $\tilde{y}_l - y_l$ symmetrically distributed about zero (inheriting the symmetry of the dense noise $w_{i,l}^m$). The conditional distribution is:
\[
\tilde{y}_l - y_l \;\Big|\; M_l \;\sim\; \text{symmetric about } 0, \quad \text{with variance } \Theta\!\left(\frac{\sigma^2 \sigma_w^2}{M_l}\right).
\]
In particular, the noise is entry-wise symmetric but not exchangeable across entries.
\end{proposition}

\begin{proof}
Conditional on $M_l$ and on the event that the median falls among uncorrupted observations, $\tilde{y}_l$ is the sample median of $M_l^{\mathrm{clean}} \geq M_l/2$ i.i.d.\ draws from a distribution symmetric about $y_l$ (since $w_{i,l}^m$ is mean-zero and symmetric by assumption). The sample median of symmetric i.i.d.\ draws is itself symmetric about the population median. Independence across $l$ follows from the independence of $w_{i,l}^m$ across entries.
\end{proof}

The entry-wise symmetry (conditional on depths) preserves some of the structure needed for Veiel's analysis, but the heterogeneous variances mean that the multilinear form coefficients, which depend on the joint distribution of the noise vector, are no longer invariant under permutations of entries. Whether this breaks or preserves the symmetry conditions for limit uniqueness depends on whether the observation heterogeneity is aligned or misaligned with the BRC structure:

\begin{itemize}[nosep]
\item If the entries involved in a potential BRC are all observed with similar depth ($M_l \approx M_{l'}$ for entries $l, l'$ appearing in the cycle constraints), the heterogeneity is irrelevant to the cycle analysis, and the results of Section~\ref{sec:rpca} carry over.
\item If entries critical to the BRC are observed at systematically different rates, for instance because on-equilibrium action pairs are observed more frequently than the off-equilibrium pairs that participate in the cycle, the asymmetric variance structure could either reinforce or break the cycle, depending on the interaction between the payoff geometry and the observation pattern.
\end{itemize}

Resolving this interaction precisely, and determining whether strategic observation generically helps or hurts cycle-breaking, requires a joint analysis of the game's equilibrium structure and the BRC geometry that goes beyond the scope of the present paper. We record it as an open question.

\begin{remark}[Connection to active learning and mechanism design]
\label{rem:active_learning}
The observation that equilibrium diversity aids payoff learning suggests a connection to the literature on active learning and information acquisition in games. A mechanism designer who can influence which equilibria are selected across markets, for instance by recommending different equilibria in different markets, could in principle \emph{design} the observation pattern to maximize coverage and minimize estimation error. This is an ``information design for estimation'' problem: the designer chooses equilibria not (only) for their payoff properties but for the information they generate about off-equilibrium payoffs. We leave the formal development of this connection for future work.
\end{remark}

%----------------------------------------------------------------------
\subsection{Summary and relationship to the full-observation case}
\label{subsec:mc_summary}
%----------------------------------------------------------------------

The three observation models form a hierarchy of increasing economic realism and decreasing statistical convenience, summarized in Table~\ref{tab:obsmodels}.

\begin{table}[t]
\centering
\caption{\textbf{Estimation rates and requirements across observation models.} Sup-norm estimation rates for the RPCA-Averaging estimator under the full-observation benchmark of Section~\ref{sec:rpca} and the three partial-observation models of Section~\ref{sec:matcomp}; ``missing entries'' refers to the uniform random model of Definition~\ref{def:entry_obs}. Rates are stated to leading order and up to a universal constant. ``Symmetry'' records whether the effective noise inherits the symmetry of the dense noise unconditionally or only conditional on realized observation depths. All four rates additionally carry the contamination bias term $\sigma \sigma_w \varepsilon$ of Proposition~\ref{prop:intermediate}(b), suppressed here; being proportional to $\sigma$, it does not affect the limit analysis.}
\label{tab:obsmodels}
\small
\setlength{\tabcolsep}{5pt}
\begin{tabular}{llll}
\toprule
Model & Rate & Symmetry & Coverage condition \\
\midrule
Full observation      & $\sigma \sigma_w \sqrt{n \log n / M}$                & Automatic   & None \\
Missing columns       & $\sigma \sigma_w \sqrt{n \log n / M_{\mathrm{obs}}}$ & Automatic   & $M_{\mathrm{obs}} \gtrsim \log n$ \\
Missing entries       & $\sigma \sigma_w \sqrt{n \log n / (pM)}$             & Automatic   & $pM \gtrsim \log n$ \\
Strategic observation & $\sigma \sigma_w \sqrt{n \log n / M_{\min}}$         & Conditional & Equilibrium diversity \\
\bottomrule
\end{tabular}
\end{table}

In all cases, the intermediate regime is preserved: the effective noise vanishes as $\sigma \to 0$ but remains strictly positive for any finite observation budget and positive noise scale. The fundamental insight is that the rank-one structure of the stacked matrix $L_0 = y \cdot \mathbf{1}_M^\top$, where every column is identical, makes partial observation far less damaging than in generic matrix completion problems. All that is needed is enough observations of each entry to compute a reliable median. The specific pattern of which observations come from which markets is irrelevant for the uniform and column-missing models, and requires only a coverage condition in the strategic case.

Combined with Theorem~\ref{thm:rank1}, the results of this section show that the estimation side of the multi-market analysis is fully robust to partial observation. The selection side extends conditionally, following the partition of Remark~\ref{rem:conditional}:

\begin{corollary}[Limit uniqueness under partial observation, conditional]
\label{cor:limuniq_mc}
Let $y \in \calS^*$ be a quasi-concave supermodular payoff with rank-one, nondegenerate (Definition~\ref{def:margin}) payoff matrices for each player. In the multi-market global game $\Gamma(\sigma, M)$ with any of the three observation models above, provided the relevant coverage and depth conditions are satisfied, the conclusions of Proposition~\ref{prop:signal_properties} continue to hold. Consequently, if Conjecture~\ref{conj:multimarket_lu} holds (equivalently, if the noise-class extension of Remark~\ref{rem:noise_class} is established), limit uniqueness holds as $\sigma \to 0$ with the observation structure held fixed.
\end{corollary}

\begin{proof}
Under any of the three models, the median signal $\tilde{y}_i$ (adapted per Definition~\ref{def:estimator_mc} as needed) delivers an effective noise $Z_i$ whose sup-norm bound carries the factor $\sigma$ (Propositions~\ref{prop:missing_columns}, \ref{prop:missing_entries}, and~\ref{prop:strategic_estimation}), so part~(a) of Proposition~\ref{prop:signal_properties} carries over. Part~(b) follows from Theorem~\ref{thm:rank1} and Corollary~\ref{cor:estimation_margin} exactly as before. For part~(c), conditional on sparse removal each median depends only on the dense noise values for its own entry, so entrywise independence and the local density lower bound are inherited from Assumption~\ref{ass:noise} as in the full-observation case. The conditional conclusion then follows from the same contagion argument that supports Conjecture~\ref{conj:multimarket_lu}.
\end{proof}

%----------------------------------------------------------------------
\section{Discussion}
\label{sec:discussion}
%----------------------------------------------------------------------

\subsection{Summary}

We have shown that low-rank payoff structure, specifically rank one, is a sufficient condition for the absence of better response cycles in supermodular games. Combined with the characterization of \citet{veiel2025limits}, this yields limit uniqueness in global games with rank-one payoffs, for any number of actions and with no genericity assumptions (Corollary~\ref{cor:limuniq}), with the dependence on his unpublished theorem recorded explicitly in Remark~\ref{rem:conditional}. The boundary of the result is sharp: rank two already cycles, no supermodular cycle has length below six, and a margin neighborhood of any nondegenerate rank-one game remains cycle-free. These structural results, and the distinction between our player-by-player rank restriction and the sum-rank notion of \citet{kannan2010rank} drawn in Remark~\ref{rem:rank_terminology}, stand unconditionally.

The multi-market RPCA extension provides both a concrete information-processing foundation and a resolution to the mismatch between the large-dimensional requirements of RPCA and the small, fixed action spaces of games. Payoff rank is a testable restriction, and linear-index and factor specifications common in applied work generate low-rank payoffs by construction, so the selection criterion identified here is one that empirical researchers can take to data.

The multi-market construction connects to the multimarket-contact literature \citep{bernheim1990multimarket}, but with an inverted mechanism: rather than multimarket contact enabling collusion through repeated-game incentives, here it enables equilibrium \emph{selection} through superior information processing.

\subsection{Open questions}

Several important questions remain open.

First, can the sign-quadrant argument of Theorem~\ref{thm:rank1} be extended to rank $r \geq 2$, perhaps via a ``partition of unity'' over singular vectors? The structural approach yields far stronger results than dimension counting for $r = 1$. Extending it to $r = 2$ or $r = 3$ would substantially broaden the economic applicability of these results, and Example~\ref{ex:rank2} shows any such extension must use more than rank and supermodularity alone.

Second, Conjecture~\ref{conj:multimarket_lu} rests on extending Veiel's characterization from his symmetric noise class to entrywise-independent noise with full local support (Remark~\ref{rem:noise_class}); establishing that extension would upgrade the conjecture, and the conditional part of Corollary~\ref{cor:limuniq_mc}, to theorems. Relatedly, Proposition~\ref{prop:sparse_breaking} suggests that multi-market RPCA approximately restores the noise symmetry required by Veiel's analysis. Quantifying this, by showing that the multilinear form coefficients converge to their symmetric values as $M \to \infty$, would extend the limit uniqueness result to games with BRCs where the RPCA residual breaks the cycle.

Third, if the number of markets $M$ is a strategic choice (or varies endogenously with the economic environment), the information-processing benefit of multi-market contact interacts with the equilibrium selection in a potentially rich way.

Fourth, with positive discounting ($\delta > 0$), the repeated interaction introduces dynamic strategic considerations (e.g., collusion, punishment strategies) that interact with the information structure. Understanding whether RPCA-based denoising amplifies or dampens equilibrium multiplicity in repeated games is an open and economically important question.

%----------------------------------------------------------------------
% Back matter
%----------------------------------------------------------------------

\section*{Acknowledgements}

The author thanks Ting Liu, Eran Shmaya, Yair Tauman, Pradeep Dubey, Mihai Manea, Sarah Betz, and Kamran Paynabar, whose course inspired the idea, as well as other attendees at both the Stony Brook Game Theory Festival and the Stony Brook Game Theory Workshop, for helpful comments and discussion. Any remaining errors are the author's own.

\section*{Declaration of competing interest}

Declarations of interest: none.

\section*{Funding}

This material is based upon work supported by the National Science Foundation under Award No.~2125295 (NRT-HDR: Detecting and Addressing Bias in Data, Humans, and Institutions). Any opinions, findings and conclusions or recommendations expressed in this material are those of the author(s) and do not necessarily reflect the views of the National Science Foundation.

\section*{Data availability}

No data were used for the research described in this article. The paper is theoretical; the numerical example of Section~\ref{sec:rankr} and the linear-programming search reported in Remark~\ref{rem:lp_lengths} are fully specified in the text and reproducible from it.

%----------------------------------------------------------------------
% References
%----------------------------------------------------------------------

\bibliographystyle{elsarticle-harv}
\bibliography{references}

\end{document}